\documentclass[journal, twocolumn]{IEEEtran}

\usepackage{presets}
\usepackage{cite}
\usepackage{amsmath,amssymb,amsfonts}
\usepackage{algorithmic}
\usepackage{graphicx}
\usepackage{textcomp}
\usepackage{xcolor}
\usepackage[left=0.9in,top=0.9in,right=0.9in,bottom=0.9in]{geometry}

\allowdisplaybreaks

\def\BibTeX{{\rm B\kern-.05em{\sc i\kern-.025em b}\kern-.08em
    T\kern-.1667em\lower.7ex\hbox{E}\kern-.125emX}}
\begin{document}

\title{Optimal Non-Adaptive Group Testing \\ with One-Sided Error Guarantees}

\author{{Daniel McMorrow} and
{Jonathan Scarlett}\thanks{The authors are with the  Department of Computer Science, School of Computing, National University of Singapore (NUS). J.~Scarlett is also with the Department of Mathematics, NUS, and the Institute of Data  Science, NUS. Emails: \url{mcmorrow@nus.edu.sg};  \url{scarlett@comp.nus.edu.sg}}
}

\maketitle

\begin{abstract}
   The group testing problem consists of determining a sparse subset of defective items from within a larger set of items via a series of tests, where each test outcome indicates whether at least one defective item is included in the test.  We study the approximate recovery setting, where the recovery criterion of the defective set is relaxed to allow some number of items (up to a certain specified threshold) to be misclassified. In particular, we consider \emph{one-sided} approximate recovery criteria, where we allow either only false negative or only false positive misclassifications.  Under false negatives only (i.e., finding a subset of defectives), we show that there exists an algorithm matching the optimal threshold of two-sided approximate recovery, albeit with exponential runtime.  Under false positives only (i.e., finding a superset of the defectives), we provide a converse bound showing that the better of two existing algorithms is optimal.
\end{abstract}



\section{Introduction}
The goal of group testing is the accurate recovery of a set of “defective” items $S \subseteq [n] \coloneqq \{1,2,\dots,n\}$ of size $\abs{S} = k$ from within a larger population via a series of tests, where each test outcome indicates whether at least one defective item is included in the test.  We wish to minimize the number of tests $T$ while still ensuring the reliable recovery of the defective set $S$. Trivially, one can test each item individually by including each item in a test, and declare the defective set to be the items returning positive tests, leading to $T = n$ tests being sufficient. However, in many cases one can reduce the number of tests required by including multiple items in each test. 
Group testing was originally introduced by Dorfman \cite{dorf} as a means to more efficiently test incoming US soldiers for syphilis during WW2, and has since seen a variety of applications such as in DNA testing \cite{curnow1998pooling, gille1991pooling}, network communications \cite{hayes1978adaptive, wolf1985born}, sparse inference \cite{malioutov2013exact} and database systems\cite{cormode2005s}.

There are several different flavors of group testing, each having different results on the minimum number of tests needed for reliable recovery. 
For instance, one can consider both \emph{adaptive} testing where previous test results can be used to inform one’s decisions in designing subsequent tests, and \emph{non-adaptive} testing where all tests are designed before any testing takes place. Moreover, one can consider \emph{noiseless} testing, where the correct outcome is always observed deterministically, or \emph{noisy} testing, where some randomness is introduced to the test outcomes (e.g., test outcomes are flipped independently with some probability). Finally, one can consider both the \emph{linear regime} where $k =\Theta(n)$ (i.e., $k$ scales linearly with $n$) or the \emph{sparse regime} where $k =\Theta(n^{\theta})$ for some $\theta\in\left(0,1\right)$ (i.e., $k$ scales sub-linearly relative to $n$).\footnote{There is also a range of $k$ such that $k = n^{1-o(1)}$ but $k = o(n)$, as considered in \cite{bay2022optimal}, but this regime is narrow and less commonly considered compared to the other two discussed here.} 

In addition, there are several recovery criteria of interest \cite{aldridge2019group}; we list some of the main ones as follows in decreasing order of how stringent they are:
\begin{itemize}
    \item In \emph{zero-error combinatorial} group testing, we require the exact recovery of any defective set $S$ of size $k$ (or in some variations, size at most $k$).  This must hold deterministically, i.e., the ``error probability'' must be zero.  This strong recovery guarantee comes at the expense of requiring at least $\min\{\Omega(k^{2}),n\}$ tests \cite{d1983bounds}.
    \item In \emph{small-error probabilistic} group testing with \emph{exact recovery}, we require that $\bP(\hat{S} \neq S) \to 0$ as $n \to \infty$, where the randomness is with respect to $S$ (chosen uniformly at random) and/or the possibly randomized test design and decoding algorithm.  This relaxation significantly reduces the required number of tests to $O(k \log n)$.
    \item One can further relax small-error probabilistic group testing to the \emph{approximate recovery} criterion, in which instead of requiring $\hat{S} = S$ exactly, we allow for a small number of false positives and/or false negatives in the estimate.  While this typically does not improve the scaling of the number of tests, it can significantly improve the constant factors.
\end{itemize}
We note that approximate recovery criteria can also be considered in the zero-error setting, and this can again reduce the required number of tests to $O(k \log n)$ \cite{cheraghchi2013noise}.



In this paper, we will focus on non-adaptive probabilistic testing in the sparse regime under various approximate recovery criteria, and we will only consider noiseless testing.

\subsection{Problem Setup} \label{sec:setup}
Given a set of items $[n] \coloneqq \{1,2,\dots, n\}$, there exists some set $S \subseteq [n]$ of defective items with $\abs{S} = k$, where $k = \Theta(n^{\theta})$ for some $\theta \in (0,1)$ and $k$ is known \emph{a priori}. The set of all possible values of $S$ is denoted by $\cS$, and $S$ is assumed to be uniformly distributed over all $n \choose k$ elements in $\cS$, which is known as the \emph{combinatorial prior}. We also consider in certain parts of our proofs the \emph{i.i.d prior} (as an intermediate step), where each item is defective independently with some probability $q \in (0,1)$.

The items included in each test are represented by a \emph{test matrix} $\sfX \in \{0,1\}^{T \times n}$, where $\sfX_{ti} = 1$ if item $i$ is included in test $t$, and $\sfX_{ti} = 0$ otherwise. 
The test outcome vector $\bfY\in \{0,1\}^T$ is generated component-wise according to the following observation model:
\begin{equation}
   \label{eq:1} Y_t = \bigvee_{i \in S}\sfX_{ti}, \ \forall t \in \{1,2,\dots, T\},
\end{equation}
where $\bigvee$ represents the \texttt{OR} operator, $Y_t$ is the $t$-th entry of $\bfY$, and $\bfX_t = (\sfX_{t1}, \sfX_{t2}, \dots, \sfX_{tn}) \in \{0,1\}^{n}$ is the $t$-th row of $\sfX$.  Given $\sfX$ and $\bfY$, a \emph{decoder} forms an estimate $\hat{S}$ of the defective set $S$.

We construct our tests \emph{non-adaptively}, meaning that all tests are designed in advance, and outcomes of tests cannot be used to influence future testing decisions.


\begin{definition} \label{def1}
We define the \emph{rate} for a given group testing algorithm aimed at identifying $k$ defective items out of a population of $n$ items as:

\begin{equation}
   \label{eq:2} R \coloneqq \frac{\log_{2}{n \choose k }}{T}.
\end{equation}
\end{definition}

We say that a rate is \emph{achievable} if for any $\delta, \epsilon>0$ and sufficiently large $n$, the number of tests $T$ satisfies $ \frac{\log_{2}{n \choose k}}{T} > R - \delta$ and the error probability incurred is at most $\epsilon$, where the notion of ``error probability'' depends on the recovery criterion. For example, for the exact recovery criterion the probability of error is given by $\bP(\hat{S} \neq S)$, whereas for the two-sided approximate recovery criterion (see Section \ref{sec:prior} for further discussion), the probability of error is given by $\bP\big(\abs{\hat{S} \setminus S}> \beta k \,\cup\, \abs{S \setminus \hat{S}} > \beta k\big)$ for some $\beta > 0$.

While previous works on approximate recovery focused mainly on this symmetric two-sided criterion, our focus in this paper is instead on the following problems, which represent the extremes of placing unequal weighting on false positives and false negatives:
\begin{itemize}
    \item Finding a near-complete subset of $S$ (a problem we will henceforth call SUBSET).
    \item Finding a slightly-enlarged superset of $S$ (a problem we will henceforth call SUPERSET).
\end{itemize}
There are natural scenarios where these problems may be more relevant than the previously considered two-sided approximate recovery criterion. For example, SUPERSET may be relevant when trying to contain a highly contagious disease in its initial stages, where a small number of healthy individuals being quarantined is worth ensuring all diseased individuals are identified.  An example where SUBSET may be relevant is identifying illegal activity via traitor tracing \cite{meerwald2011group}, where it is paramount to avoid false accusations.

Note that exact recovery solves both SUBSET and SUPERSET, since after successful recovery, one can either remove defective items or add non-defective items to $S$. One can also consider the asymmetric approximate recovery case where we allow at most $\alpha_{1}k$ false negatives and $\alpha_{2}k$ false positives for general fixed values of $\alpha_{1}, \alpha_{2} > 0$. However, this problem is a relatively straightforward extension of existing work \cite{scarlett2017how, truong2020all}; see Appendix \ref{sec:asymmetric} for details.

For the SUBSET and SUPERSET problems, the error probabilities are given as follows with parameters $\eta^-, \eta^+ > 0$. 

\begin{definition} \label{def:subset}
    For the SUBSET problem with parameter $\eta^- \in [0,1]$, the error probability is defined by
    \begin{equation}
        \label{eq:4} P_{\mathrm{e}}^{-} \coloneqq \bP\left(\hat{S}\not \subseteq S  \right)
    \end{equation}
for an estimate $\hat{S}$ of $S$ with $\abs{\hat{S}} = (1-\eta^-)k$.
\end{definition}

\begin{definition} \label{def:superset}
    For the SUPERSET problem with parameter $\eta^+ \ge 0$, the error probability is defined by 
    \begin{equation}
        \label{eq:5} P_{\mathrm{e}}^{+} \coloneqq \bP\left(\hat{S} \not \supseteq S\right)
    \end{equation}
for an estimate $\hat{S}$ of $S$ with $\abs{\hat{S}} = (1+\eta^+)k$.
\end{definition}

Note that the probabilities in these definitions are over the random defective set $S$ and any potential randomness in the test design and decoder.  Moreover, we assume for notational convenience that $(1-\eta^-)k$ and $(1+\eta^+)k$ are integers, but if not, all of our results will hold regardless of whether these values are rounded up or rounded down.

\subsection{Prior Work} \label{sec:prior}

{\bf Exact recovery.} Dorfman's original scheme \cite{dorf} was a two-stage adaptive algorithm consisting of splitting up the items into disjoint groups, and including all items of each group in a single test. All items in negative tests are marked as non-defective, while the remaining items in positive tests are subsequently tested individually. While this does indeed beat individual testing, such a scheme is highly suboptimal when the number of defectives is sufficiently sparse relative to the number of items.

In the sparse setting (namely, $k = \Theta(n^{\theta})$ with $\theta \in (0,1)$), there exists a ubiquitous lower bound of $ T \geq \big( k\log_{2}\frac{n}{k}\big)(1-o(1))$ on the number of tests required to obtain a vanishing probability of error, regardless of adaptivity or non-adaptivity. This bound is known as the counting bound \cite{baldassini2013capacity}. The intuition behind this bound stems from the fact that $\log_{2}{n \choose k} = \big(k \log_{2}\frac{n}{k}\big)\big(1+o(1)\big)$ bits are required to describe the defective set, and each test can provide at most 1 bit of information since the test outcomes are binary. This intuition can easily be made precise using standard information theoretic tools such as Fano's inequality \cite[Theorem 1]{chan2011non}, which implies we cannot obtain an arbitrarily small probability of error with fewer tests. This result has also been strengthened to obtain a strong converse bound stating that $\bP(\hat{S} \neq S) \to 1$ as $n \to \infty$ when $T \leq \left(k\log_{2}\frac{n}{k}\right)(1-\epsilon)$ for arbitrarily small $\epsilon > 0$ \cite[Theorem~3.1]{baldassini2013capacity}. 

The counting bound can be met when adaptive testing is used for all $\theta \in (0,1)$ using a generalized binary splitting algorithm due to Hwang \cite{hwang1972method}.
The counting bound can also be achieved for $\theta \in \big(0, \frac{\log 2}{1+\log 2}\big) \approx (0,0.409)$ non-adaptively, whereas a converse bound specific to the non-adaptive setting shows that more tests are required for $\theta \geq \frac{\log 2}{1+\log 2}$ \cite[Theorem~1.1]{coja2020optimal}. Furthermore, \cite[Theorem~1.2]{coja2020optimal} shows that there exists a polynomial-time algorithm matching this bound in the non-adaptive setting.  (Prior to their work, the bound was showed to be information-theoretically achievable in \cite{coja2020information}.)

{\bf Approximate recovery.} There is also a converse bound specific to the approximate recovery setting \cite[Theorem 3]{scarlett2016phase}, which states that the high probability recovery of a set $\hat{S}$ such that $\max\big\{\abs{\hat{S} \setminus S}, \abs{S \setminus \hat{S}}\big\} \leq \beta k$ for fixed $\beta \in (0,1)$ is not possible when the number of tests satisfies $T \leq (1-\beta)\left(k \log_{2}\frac{n}{k}\right)(1-\epsilon)$ for arbitrarily small $\epsilon > 0$. 
A matching upper bound is also known \cite{truong2020all}, and the polynomial-time algorithm from \cite{coja2020optimal} can be used to achieve the same bound. The techniques in these works appear to be heavily reliant on the approximate recovery criterion being two-sided, and thus not directly suitable for SUBSET and SUPERSET.  For instance, in \cite{scarlett2016phase} the decoder specifically searches over size-$k$ estimates, which immediately implies an equal number of false positives and false negatives.

We briefly pause to highlight the test designs used in these works:
\begin{itemize}
    \item In \cite{scarlett2016phase,truong2020all} the {\em Bernoulli design} is used, in which each item is placed in each test independently with probability $\frac{\nu}{k}$ for some $\nu > 0$.
    \item In \cite{coja2020information} (see also \cite{johnson2018performance}) the \emph{near-constant tests-per-item design} (or \emph{near-constant column weight design}) is used, in which each item is placed in $L = \frac{\nu T}{k}$ tests (for some $\nu > 0$) chosen uniformly at random with replacement.
    \item The polynomial-time method in \cite{coja2020optimal} is based on a \emph{spatially coupled} test design that builds on the near-constant tests-per-item design but is specifically geared towards lowering the computation required for decoding.
\end{itemize}

There are also some related works tackling similar problems to SUBSET and SUPERSET.  Aldridge \cite{aldridge2021pooled} studied a related problem of isolating defectives, and translated to our setup, one of their results solves an instance of the SUBSET problem using a combination of the exact recovery algorithm in \cite{coja2020optimal} and a proposed algorithm using ideas from SAFFRON \cite{lee2019saffron}.  However, in the regimes where the latter algorithm attains the better rate, it only recovers roughly $58\%$ of the defective set and attains a rate below $0.19$, which turns out to be highly suboptimal. Sharma et al.~\cite{sharma2013finding} also considered the problem of finding a subset of healthy individuals in the population, but taking the complement of that set would give a ``very enlarged superset'' rather than ``slightly enlarged'' as we seek. The SUPERSET problem is also related to list decoding, which has been considered in works such as \cite{d2015almost, gilbert2008group, cheraghchi2013noise, indyk2010efficiently} but with very different goals to ours (e.g., as a stepping stone to exact recovery in the zero-error setting).  Mézard et al. \cite{mezard2011group} established a converse bound stating that a rate above $\log 2$ is not achievable for a problem related to SUPERSET. However, their result concerns two-stage adaptive procedures rather than the non-adaptive designs we consider, and they are interested in the zero-error recovery problem rather than the small-error criterion that we consider. 

\subsection{Achievable Rates Adapted from Prior Works}

With the definitions from Section \ref{sec:setup} in mind, we now state the best known existing results for both problems.  The relevant existing works stated these results for $\eta^-,\eta^+ = \Theta(1)$ in Definition \ref{def:subset}, but by slightly adapting their proofs, we produce a stronger result in which these parameters can decay to zero at a certain (sufficiently slow) speed.

\begin{theorem}[Adapted from \cite{aldridge2014group, chan2011non, aldridge2019group}] \label{thm1}
    Let $\theta \in (0,1)$ and

    \begin{equation}
        \label{eq:6} R^{*} = \max\{\zeta(\theta), \log 2\},
    \end{equation}
    where

    \begin{equation}
       \label{eq:7} \zeta(\theta) \coloneqq \min\left\{1, \log 2 \frac{1-\theta}{\theta}\right\}.
    \end{equation}
Then if $\eta^- = \Omega(k^{-\lambda(\frac{1}{\theta}-1)})$ for any $\lambda \in [0,1)$, and $\eta^+ = k^{-o(1)}$, the rate $R^{*}$ is achievable in both the SUBSET and SUPERSET problems.
\end{theorem}

The condition $\eta^+ = k^{-o(1)}$ covers a range of possible values, including not only $\eta^+ = \Theta(1)$ but also the case of decaying sufficiently slowly (e.g., as $(\log k)^{-c}$ for some $c > 0$). Moreover, the condition on $\eta^-$ allows not only $\eta^- = \Theta(1)$ and ``slow'' decay to zero at rate $k^{-o(1)}$, but also faster (polynomial) decay; by letting $\lambda$ be close to $1$ we can allow $\eta^- = \big( \frac{k}{n} \big)^{1-\epsilon}$ for arbitrarily small $\epsilon >0$, amounting to at most $k \big( \frac{k}{n} \big)^{1-\epsilon}$ false negatives in the SUBSET problem. 
Note also that the term $\zeta(\theta)$ is the optimal exact recovery rate (among all designs) from \cite{coja2020optimal}, and is present because achieving exact recovery implies solving both problems even with $\eta^+, \eta^-=0$.  This term is dominant for $\theta < \frac{1}{2}$, and for any such $\theta$ the quantity $\eta^-k = \Theta(k(\frac{k}{n})^{\lambda})$ is $o(1)$ for $\lambda$ sufficiently close to one, which means that the ``approximate recovery'' guarantee of SUBSET is in fact equivalent to exact recovery.

For the SUBSET problem, the rate $\log 2 \approx 0.693$ is achievable using the \emph{Definite Defectives} (DD) algorithm along with the near-constant tests-per-item design.  See \cite{aldridge2014group} for the algorithm studied under Bernoulli designs, \cite{johnson2018performance} for its study under the improved design, and \cite[Section 5.1]{aldridge2019group} for an overview of the extension to approximate recovery.

For the SUPERSET problem, the rate of $\log 2 $ is achievable using the \emph{Combinatorial Orthogonal Matching Pursuit} (COMP) algorithm.  See \cite{chan2011non} for the algorithm and Bernoulli designs, \cite{johnson2018performance} for its study under the improved test design, and \cite[Section 5.1]{aldridge2019group} for an overview of the extension to approximate recovery. 

Since the details on these extensions are omitted from \cite[Section 5.1]{aldridge2019group}, and since the outline therein only covers the case that $\eta^-,\eta^+ = \Theta(1)$, we detail the proof of Theorem \ref{thm1} in Appendix \ref{sec:comp_dd}.

\subsection{Contributions}\label{AA}

We now state our main contributions in this paper:
\begin{itemize}
    \item We provide a new algorithm for the SUBSET problem that achieves a rate of 1 for all $\theta \in (0,1)$ when $\eta^- = \Omega(k^{-\lambda(\frac{1}{\theta}-1)})$ for any $\lambda \in [0,\frac{1}{2})$ (albeit with exponential runtime), thus matching the counting bound and beating the best existing achievability result.
    \item We provide a strong converse bound for the SUPERSET problem that shows the rate in \eqref{eq:6} is optimal for any $\eta^+ = k^{o(1)}$ among arbitrary non-adaptive designs, with this converse holding for even larger $\eta^+$ when $\theta \in (0.409,\frac{1}{2})$, namely, $\eta^+ = O(k^{\lambda})$ with $\lambda \in [0,\frac{1}{\theta}-2]$.
\end{itemize}
A more detailed overview is given in Section \ref{sec:results}.

\subsection{Notation} \label{sec:notation}

We generally use boldface letters to denote vectors, lower-case to denote scalars, and capital letters to denote random variables (though the number of tests $T$ is an exception). 
We use standard Bachmann-Landau notation $o(\cdot), \ \omega(\cdot), \ O(\cdot), \ \Omega(\cdot), \ \Theta(\cdot)$ to denote asymptotic relations. For a matrix $\sfM$ (resp. vector $\bfv$) and set $\cJ \subseteq [n]$, $\sfM_{\cJ}$ denotes the \emph{sub-matrix} of $\sfM$ (resp. sub-vector of $\bfv$) corresponding to the columns (resp. entries) indexed by $\cJ$. If $\cJ = \emptyset$ we use the convention that $\sfM_{\emptyset}$ (resp. $\bfv_{\emptyset}$) $ = \emptyset$. For a set $\cN$ and some $\cU \subseteq \cN$, we denote $\cU^\mathsf{c} \coloneqq \cN \setminus \cU$ as the complement of $\cU$. All logarithms have base $e$ unless specified otherwise (e.g., via $\log_{2}$).



The Hamming distance between two vectors $\bfu, \bfv \in \{0,1\}^{n}$ is given by $\dH(\bfu, \bfv)$. We also slightly abuse notation throughout and use $\dH(S, S')$ for $S, S' \subseteq [n]$, which is to be understood as the Hamming distance between the binary vectors having $1$'s in the locations indexed by $S, S'$ and $0$'s elsewhere.

\begin{figure}
    \centering 
    \includegraphics[width=0.84\columnwidth]{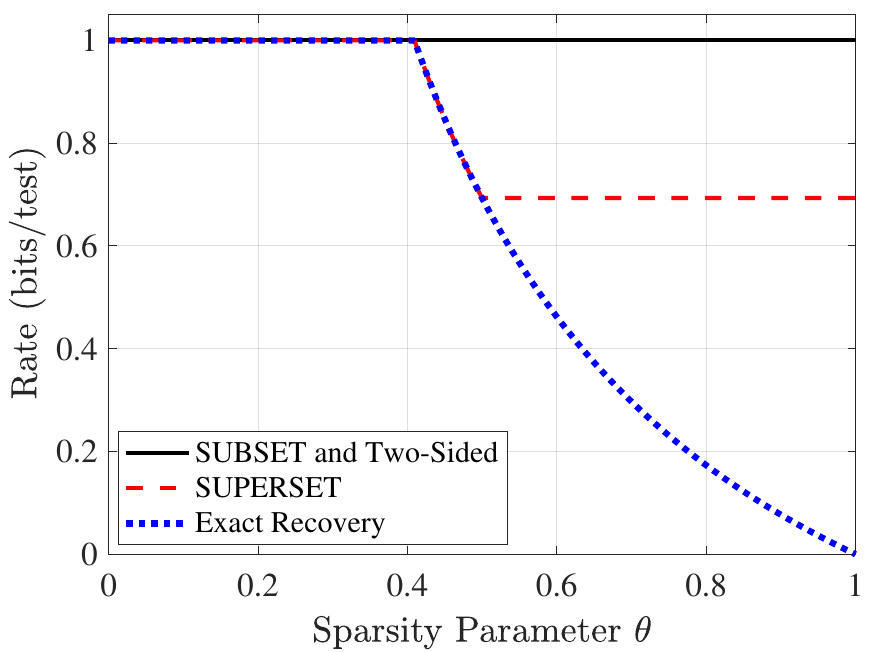}
    
    \caption{Summary of optimal rates under various recovery guarantees, where approximate recovery is with $o(k)$ misclassifications (e.g., requiring $\eta^-,\eta^+=o(1)$).  The result for exact recovery is from \cite{coja2020optimal}, and the result for two-sided approximate recovery is from \cite{scarlett2016phase}. } 
    \label{fig:rates}
\end{figure}

\section{Main Results} \label{sec:results}

We now formally state our main results; these are summarized in Figure \ref{fig:rates} along with existing results for exact recovery and two-sided approximate recovery. The proofs of these results are provided in Section \ref{sec:proofs}.

\begin{theorem} \label{thm2}
    In the SUBSET problem, for any $\theta \in (0,1)$, there exists an algorithm that achieves a rate of 1 for recovering an estimate $\hat{S}$ of size $(1 - \eta^-)k$, when $\eta^- = \Omega(k^{-\lambda(\frac{1}{\theta}-1)})$ for any $\lambda \in [0,\frac{1}{2})$.
\end{theorem}

\begin{theorem} \label{thm3}
    In the SUPERSET problem, for any $\theta \in (0,1)$ and $\eta^+ = k^{o(1)}$, any rate above
    \begin{equation}
       \label{eq:10} R^{*} = \max\{\zeta (\theta), \log 2\} 
    \end{equation}
where
    \begin{equation}
        \label{eq:11} \zeta(\theta) \coloneqq \min \left\{1, \log 2 \frac{1 - \theta}{\theta}\right\}
    \end{equation}
cannot be achieved regardless of test design. In particular, for rates above $R^{*}$, it holds that $P_{\mathrm{e}}^{+} \to 1$. Moreover, when $\theta \in \big(\frac{\log 2}{1+\log 2}, \frac{1}{2}\big)$ the same result also holds for $\eta^+ = O(k^{\lambda})$, for any $\lambda \in [0,\frac{1}{\theta}-2]$.
\end{theorem}

While the results in Theorems \ref{thm2} and \ref{thm3} establish the optimal rate for SUBSET and SUPERSET in the specified parameter regimes, there are more fine-grained notions of optimality that remain unsolved. In particular, it is unclear whether the range of $\lambda$ (and hence $\eta^-, \eta^+$) in Theorems \ref{thm2} and \ref{thm3} can be improved, or if there exists a polynomial-time algorithm for the SUBSET problem that can achieve the same guarantees of that in Theorem \ref{thm2} (recall that our algorithm has exponential runtime). These questions are further discussed in the following remarks.

\begin{remark} \label{rem:lambda}
{\em (Discussion on $\eta^-, \eta^+$)}
While the range of $\lambda$ in Theorems \ref{thm2} and \ref{thm3} is not our main focus, we note that the range in Theorem \ref{thm2} permits $\eta^- = \big(\sqrt{\frac{k}{n}}\big)^{1-\epsilon}$ (i.e., $k\big(\sqrt{\frac{k}{n}}\big)^{1-\epsilon}$ false negatives) for arbitrarily small $\epsilon > 0$, which is a stronger result than $\eta^- = \Theta(1)$ or $\eta^- = k^{-o(1)}$.  Similarly, according to the range in Theorem \ref{thm3}, when $\theta \in \big(\frac{\log 2}{1+\log 2}, \frac{1}{2}\big)$ the converse result holds even when $\eta^+ = \big( \frac{n}{k^2} \big)^{1-\epsilon}$ (i.e., $k\big( \frac{n}{k^2} \big)^{1-\epsilon} \gg k$ false positives).

While we attempted to find the smallest possible values of $\eta^-, \eta^+$ in Theorems \ref{thm2} and \ref{thm3} based on what our own analysis can give, we do not claim these to be the best possible, especially when considering arbitrary non-adaptive designs.  In particular, we expect that adopting a near-constant column weight design for SUBSET would lead to smaller values of $\eta^-$ being allowed. One way of seeing this is to observe that setting $\theta < \frac{1}{3}$ and $\lambda \approx \frac{1}{2}$ in Theorem \ref{thm2} permits $\eta^- < \frac{1}{k}$, which is equivalent to exact recovery.  This is consistent with exact recovery being possible for all $\theta \le \frac{1}{3}$ under the Bernoulli design.  For the near-constant column weight design, this range increases to $\theta \le \frac{\log 2}{1+\log 2} \approx 0.409$, which suggests that the range of allowed $\lambda$ may be higher, namely $\lambda \in [0,\log 2)$.

Despite being potentially suboptimal in this sense, we stick to the Bernoulli design since it simplifies the analysis, and because making $\eta^-, \eta^+$ as small as possible is not our main focus.  
\end{remark}

\begin{remark} {\em (Discussion on Computation)}
    The optimal threshold for SUPERSET can be achieved with polynomial-time decoding. When the rate matches that of exact recovery, this follows directly from the results of \cite{coja2020optimal} based on spatial coupling.  The only other case is that in which the rate is $\log 2$, which is achieved by the Definitive Defectives (DD) algorithm (Theorem \ref{thm1}).  
    For the SUBSET problem, however, our achievability proof (Theorem \ref{thm2}) is purely information-theoretic and uses a highly inefficient decoder, since the first step of our algorithm uses a decoder from \cite{scarlett2016phase} that searches over all ${n \choose k} = e^{\Theta(k\log n)}$ subsets of size $k$, and our second step involves a similar kind of search over $\Omega\big({k \choose 2\eta^- k}{n \choose \eta^-k}\big) = n^{\omega(1)}$ sets. We leave it as an open problem as to whether the same can be achieved in polynomial time, possibly via the spatial coupling approach of \cite{coja2020optimal}; see Section \ref{sec:conclusion} for further discussion.
\end{remark}

We can also establish the tightness of these results by comparing them to existing results:
\begin{itemize}
    \item In Theorem \ref{thm2} the rate is 1, and it is known from \cite{truong2020all} that this cannot be improved even when we make two further relaxations: (i) only require $\eta^- = o(1)$ instead of the smaller value from Theorem \ref{thm2}; (ii) additionally allow $o(k)$ false positives instead of allowing none (i.e., two-sided approximate recovery).
    \item In Theorem \ref{thm3}, we match the upper bound from Theorem \ref{thm1} for all $\eta^+ = k^{o(1)}$.
\end{itemize}
Based on these findings and Figure \ref{fig:rates}, we see that the order of difficulty of these problems from ``hardest'' to ``easiest'' is (i) exact recovery, (ii) SUPERSET, (iii) SUBSET, (iv) two-sided approximate recovery; strict inequality in the optimal rate is possible between (i)-(ii) and (ii)--(iii), whereas (iii)-(iv) in fact have the same optimal rate of 1, at least for the range of $\eta^-$ stated in Theorem \ref{thm2}.

To give some intuition as to why SUPERSET is ``harder'' than SUBSET, we consider the number of ``masked'' items, whose status is ambiguous due to only appearing in tests with other defectives (see Definition \ref{def:masked} below for a precise definition):
\begin{itemize}
    \item With high probability, for rates above $\log 2$ there are a large number ($\omega(k)$) of masked non-defective items, whereas for rates below $\log 2$ there are relatively few ($o(k)$).
    \item With high probability, for rates in $(\zeta(\theta),1)$ (where we recall that $\zeta(\theta)$ is the exact recovery rate), there are some (at least $1$ but $o(k)$) masked defective items, whereas for rates below $\zeta(\theta)$ there are none.
    \item For rates above 1, estimates very far from the true defective set become consistent with the test results, so none of the approximate recovery criteria can be attained.
    \item SUPERSET is solvable when there are \emph{either} relatively few masked non-defective items (so they can be guessed as being defective without significantly increasing the size of the estimate), \emph{or} there are no masked defectives and the rate is below 1 (meaning we can achieve exact recovery and thus solve SUPERSET). 
    \item In contrast, SUBSET is solvable for rates all the way up to $1$, since the $o(k)$ masked defectives can be guessed as being non-defective without significantly impacting the size of the estimate, and even the presence of $\omega(k)$ masked non-defectives does not impact this.
\end{itemize}
Note that the above claims on when masking occurs will be evident in our proofs, and can also be inferred from \cite{coja2020optimal}.

For SUPERSET, the proof of Theorem \ref{thm3} is directly based on formalizing the above intuition, with masking playing a central role.  In contrast, for SUBSET, Theorem \ref{thm2} is an achievability result, so a concrete algorithm is needed.  Our high-level approach is to (i) form an initial two-sided approximate recovery estimate $\hat{S}'$ using a Bernoulli test design \cite{scarlett2016phase}, and (ii) using the same tests and their outcomes, perform a ``local search'' of size-$(1-\eta^-)k$ sets close to $\hat{S}'$, seeking one that ``explains'' as many positive tests as possible.  With step (i) coming from prior work (with minor modifications), the bulk of our analysis is concerned with step (ii), which departs significantly from existing work.

A further consequence of Theorem \ref{thm2} is when we only require recovering a fraction $1-\alpha$ of the defectives for a fixed constant $\alpha \in (0,1)$, we can achieve a rate of $\frac{1}{1-\alpha}$.  Formally, we have the following.

\begin{corollary} \label{corollary1}
    For $\epsilon > 0$ and 
    any fixed $\alpha \in (0,1)$, there exists an algorithm that outputs an estimate $\hat{S}$ of $S$ of size $(1 - \alpha)k$ such that $\hat{S} \subseteq S$ with probability $1-o(1)$ whenever $T \geq (1-\alpha)(k\log_{2}\frac{n}{k})(1+\epsilon)$. In other words, a rate of $\frac{1}{1-\alpha}$ is achievable.
\end{corollary}
This is achieved using a technique from \cite{truong2020all}, which ``ignores'' roughly a fraction $\alpha$ of the items (declaring them non-defective) and then applies approximate recovery on the remaining items.  We provide the details in Section \ref{sec:cor_proof}.

\section{Proofs} \label{sec:proofs}

Here we provide the proofs of Theorem \ref{thm2}, Corollary \ref{corollary1}, and Theorem \ref{thm3}, including a description of the algorithm mentioned in Theorem \ref{thm2}. We drop the superscripts from $\eta^-, \eta^+$ throughout the proofs for ease of presentation; whether $\eta$ represents $\eta^-$ or $\eta^+$ will be clear from the context.

\subsection{Proof of Theorem \ref{thm2} (Achievability for SUBSET)} \label{sec:ach_subset_pf}

\medskip
\noindent \textbf{\underline{Preliminaries and Roadmap}}
\medskip

\noindent Since we are seeking to establish a rate of 1, we assume throughout the proof that the number of tests is $T = \big(k\log_{2}\frac{n}{k}\big)(1+\epsilon)$ for arbitrarily small $\epsilon > 0$.  We also implicitly condition on a fixed choice of $S$ (say $s := \{1,\dotsc,k\}$, though we often still adopt the original notation $S$); the analysis will hold for any such choice, with the randomness only being over the test design (whose distribution is described below).

Before presenting the algorithm for SUBSET, we first briefly describe the approximate recovery algorithm from \cite{scarlett2016phase}, the output of which is used in our algorithm.  
For a given $\beta \in (0,1)$ the algorithm from \cite{scarlett2016phase} seeks to return an estimate $\hat{S}'$ of size $k$ such that the following condition holds:
\begin{equation}
    \max\big\{\abs{S \setminus \hat{S}'}, \abs{\hat{S}' \setminus S}\big\} \leq \beta k \iff \dH(S, \hat{S}') \leq 2\beta k, \label{eq:dH_equiv}
\end{equation}
where the equivalence holds due to both sets having the same size. (Recall that the Hamming distance $\dH$ between sets is interpreted as that between the associated length-$n$ indicator vectors.)

The tests follow a \emph{Bernoulli test design}, where each item is included in each test independently with some probability $p$. The value of $p$ is set to $p = \frac{\log 2}{k}\big(1+o(1)\big)$ such that each test is positive with probability $\frac{1}{2}$, thus maximizing the entropy of each test. The approximate recovery algorithm then checks through every $\bar{s} \in \cS$, and looks for a unique choice such that, for all partitions $(\sdif, \seq)$ of $\bar{s}$ satisfying $\abs{\sdif} > \beta k$, the following holds:
\begin{equation} \label{eq:12}
    \imath^{T}\left(\sfX_{\sdif} ; \bfy \mid  \sfX_{\seq}\right) > \gamma_{\abs{\sdif}},
\end{equation}
where $\gamma_{\lfloor\beta k + 1\rfloor}, \gamma_{\lfloor \beta k + 2\rfloor}, \dots, \gamma_{k}$ are positive constants and $\imath^{T}\left(\sfX_{\sdif} ; \bfy \mid  \sfX_{\seq}\right)$ is a \emph{conditional information density} (the precise definition is omitted as it is not important for our purposes).
%
%
The proof of \cite[Theorem 3]{scarlett2016phase} reveals that for appropriate choices of $\gamma_{\lfloor\beta k + 1\rfloor}, \gamma_{\lfloor \beta k + 2\rfloor}, \dots, \gamma_{k}$ the algorithm succeeds in finding such a set if $\beta \in (0,1)$ is fixed and $T \geq \left(k\log_{2}\frac{n}{k}\right)(1+\epsilon)$ for arbitrarily small $\epsilon > 0$. It turns out that their analysis can also be extended from constant $\beta$ to any $\beta = \Omega(k^{-\lambda (\frac{1}{\theta}-1)})$ with $\lambda \in [0,\frac{1}{2})$, with only minor changes.  The proof of this slightly strengthened statement is outlined in Appendix \ref{sec:two_sided}.

We also present one further definition that will be used in the description of our algorithm.

\begin{definition} \label{def:explained}
    A test $t \in \{1,2,\dots, T\}$ is \emph{explained} by an item $i \in [n]$ if it holds that both $\sfX_{ti} = 1$ and $y_{t} = 1$, and there is no $t' \in \{1,2,\dots, T\}$ such that $\sfX_{t'i} = 1$ and $y_{t'} = 0$. Moreover, a test $t \in \{1,2,\dots, T\}$ is explained by a set $\cJ \subseteq [n]$ if $t$ is explained by at least one $i \in \cJ$.
\end{definition}

Put simply, this means that a test is explained by an item if the test is positive and contains that item, and we do not know for certain that the item is non-defective. With this definition in mind, we now describe the SUBSET algorithm; see Algorithm \ref{alg2} for the pseudo-code.

\begin{algorithm}[!t]
    \caption{SUBSET Algorithm \label{alg2}}

    \textbf{Input:} Estimate $\hat{S}'$ from the approximate recovery algorithm \cite{scarlett2016phase} with $\beta=\eta$, test matrix $\sfX$, test outcome vector $\bfy$.

    \textbf{Output:} An estimate $\hat{S}$ of $S$.

    \begin{algorithmic}[1]
        \STATE Initialize $T_{\mathrm{best}} = 0$
        \FOR{all $\bar{S}$ such that $\dH(\hat{S}',\bar{S}) \le 3\eta k$ and $\abs{\bar{S}}=(1-\eta)k$}
        \STATE Set $T_{\mathrm{exp}}(\bar{S}) = (\#$tests explained by items in $\bar{S})$
        \IF{$T_{\mathrm{exp}}(\bar{S}) > T_{\mathrm{best}}$}
        \STATE Set $T_{\mathrm{best}} \leftarrow T_{\mathrm{exp}}(\bar{S})$ and $S_{\mathrm{best}} \leftarrow \bar{S}$
        \ENDIF
        \ENDFOR
        \RETURN{$\hat{S}=S_{\mathrm{best}}$}
    \end{algorithmic}
\end{algorithm}

The SUBSET algorithm first takes the estimate returned from the approximate recovery algorithm (which has size $k$), and then searches through all subsets of size $(1 - \eta)k$ within Hamming distance at most $3\eta k$ from the approximate recovery estimate, returning the set that explains the most tests. Note that this algorithm does not conduct any additional tests outside of the ones conducted to obtain $\hat{S}'$. As such, the testing matrix $\sfX$ follows a Bernoulli test design with the following inclusion probability:
\begin{equation}
    p = \frac{\log 2}{k}\big(1+o(1)\big) ~\text{ satisfying }~ \bP(Y_t = 1) = \frac{1}{2}. \label{eq:p_choice}
\end{equation}

The algorithm is successful if the returned set $\hat{S}$ is a subset of the defective set $S$. We will prove that this is indeed the case with high probability. To show this, it suffices to show that a subset of $S$ explains the most tests with high probability among all sets of size $(1 - \eta)k$ within distance $3\eta k$ of the approximate recovery estimate $\hat{S}'$. We will show this via the following high-level roadmap:
\begin{enumerate}
    \item We first show the high probability existence of a size-$(1-\eta)k$ subset $\tilde{S} \subset S$ such that each item is good with respect to many tests.
    \item We then establish two useful facts: (i) Any set $\tilde{S}' \not \subset S$ considered by the SUBSET algorithm can be formed from $\tilde{S}$ by ``swapping out'' $\ell$ of its items for $\ell$ non-defectives, for a suitable value of $\ell$; (ii) In order to have a ``bad event'' of $\tilde{S}$ explaining fewer tests than some such $\tilde{S}' \not \subset S$, it is necessary that the number of tests that are ``no longer explained'' (upon removing $\ell$ items from $\tilde{S}$) is exceeded by the number that are ``newly explained'' (upon adding $\ell$ items to the reduced set of size $|\tilde{S}|-\ell$).
    \item Finally, we show that the ``goodness'' property of $\tilde{S}$ outlined in Step 1 ensures that the event just described has sufficiently sufficiently small probability (even after a union bound over $\tilde{S}'$ and $\ell$) via a careful argument based on binomial concentration.
\end{enumerate}

\medskip
\noindent \textbf{\underline{Detailed Proof}}
\medskip

\noindent We proceed to formalize the above outline. A first thought would be to condition on the approximate recovery estimate $\hat{S}'$ and analyze the sets around this estimate. However, this approach appears to be difficult, since conditioning on this event means that the entries of $\sfX$ are no longer independent, which complicates the analysis.

Instead, we study sets $S'$ such that $\dH(S, S') \leq 5\eta k$, i.e., sets that lie within
\begin{equation}
    \label{eq:S_an} \cS_{\mathrm{an}} \coloneqq \{S' \subseteq [n] : \abs{S'} = (1- \eta)k, \ \dH(S, S') \leq 5\eta k\}.
\end{equation}
(The subscript ``an'' indicates being used for mathematical ``analysis''.)
By comparison, the sets considered by the algorithm are as follows:
\begin{equation}
    \label{eq:S_alg} \cS_{\mathrm{alg}} \coloneqq \{\bar{S} \subseteq [n] : \abs{\bar{S}} = (1-\eta)k, \ \dH(\hat{S}', \bar{S}) \leq 3\eta k\}.
\end{equation}
We know that $\dH(\hat{S}',S) \le 2\eta k$ (with high probability) due to the choice $\beta = \eta$ in Algorithm \ref{alg2}, and when this holds, the triangle inequality gives $\cS_{\mathrm{alg}} \subseteq \cS_{\mathrm{an}}$.  In other words, if we prove a certain property for $\cS_{\mathrm{an}}$, we can transfer that property to $\cS_{\mathrm{alg}}$.  When studying $\cS_{\mathrm{an}}$, there is no need for conditioning of the kind discussed above.


Thus, we seek to show that some subset $\tilde{S} \subseteq S$ of size $(1-\eta)k$ explains the most tests among those in $\cS_{\mathrm{an}}$, while ensuring that the choice of $\tilde{S}$ still satisfies $\tilde{S} \in \cS_{\mathrm{alg}}$.  To do this, we consider making ``local changes'' to this suitably-chosen $\tilde{S}$, and show that with high probability, making these local changes reduces the number of tests explained. Before we do so, we introduce the following definition.

\begin{definition} \label{def5} 
    A test $t \in \{1,2,\dots, T\}$ is \emph{good} with respect to a defective item $i \in S$ if $i$ appears in $t$ and no other defective items appear in $t$.
\end{definition}

This definition is relevant since when replacing an item $i \in \tilde{S}$ by some other non-defective item, these ``good'' tests are no longer explained by defective items in the perturbed set. For a given $i \in S$, a test $t \in \{1,2,\dots, T\}$ is good with respect to $i$ with probability $\frac{\log 2 (1+o(1))}{k}(1 - \frac{\log 2 (1+o(1))}{k})^{k-1} = \frac{\log 2}{2k}\big(1+o(1)\big)$ (see \eqref{eq:p_choice}). Hence, the number of tests that are good with respect to $i$ is distributed as $G_{i} \sim \bin\big(T, \frac{\log 2}{2k}(1+o(1))\big)$. The following lemma gives a high probability lower bound on the number of items $i \in S$ such that $G_{i}$ is sufficiently large. 

\begin{lemma} \label{lem1}
    If $T = \big(k\log_{2}\frac{n}{k}\big)(1+\epsilon)$ for some $\epsilon > 0$, and $\eta = \Omega((\frac{k}{n})^{\lambda})$ for some $\lambda \in \big[0,\frac{1}{2}\big)$,\footnote{We note that $k^{-\lambda(\frac{1}{\theta}-1)} = \Theta((\frac{k}{n})^\lambda)$; the former parametrization is mainly used for the theorem statement while the latter is more convenient here and in subsequent proofs.} then for any $\delta_{1} = 1 - o(1)$, it holds with probability $1 - o(1)$ that there are at least $k(1 - \eta)$ defective items $i$ such that $G_{i} \geq \frac{1 - \delta_{1}}{2}\log \frac{n}{k}$.
\end{lemma}
\begin{proof}
For a given $\smash{i \in S}$, the above calculations show that
\begin{align}
    \mu_{G} &\coloneqq\bE G_{i} \\
    &= \frac{T \log 2}{2k}(1+o(1)) \\
    &=\frac{(1 + \epsilon)\log 2}{2}\log_{2}\frac{n}{k}(1+o(1)) \\
    &= \frac{1 + \epsilon}{2}\log \frac{n}{k}\big(1+o(1)\big).
\end{align}
Hence, by standard binomial concentration bounds (Appendix \ref{app:bino}) and the fact that $\mu_G \geq \frac{1}{2} \log \frac{n}{k}$, it holds that
\begin{align}
    \bP\Big(G_{i} &\leq \frac{1-\delta_{1}}{2}\log\frac{n}{k}\Big) \leq \bP(G_{i} \leq (1-\delta_{1})\mu_{G})  \\
    \label{eq:gi_bound_1} &\leq \exp\left(-\big((1-\delta_{1})\log(1-\delta_{1}) + \delta_{1}\big)\mu_{G}\right) \\
    \label{eq:gi_bound_2} &\le \exp\left(-\frac{1}{2}\big((1-\delta_{1})\log(1-\delta_{1}) + \delta_{1}\big)\log \frac{n}{k} \right) \\
    \label{eq:gi_bound_3} &= \left(\frac{k}{n}\right)^{\frac{(1-\delta_{1})\log(1-\delta_{1}) + \delta_{1}}{2}} \\
    \label{eq:gi_bound_4} &= \left(\frac{k}{n}\right)^{\frac{1}{2}- o(1)},
\end{align}
where \eqref{eq:gi_bound_4} follows since $(1-\delta_{1})\log(1-\delta_{1}) + \delta_{1} = 1 - o(1)$ when $\delta_{1} = 1-o(1)$. Let 
\begin{equation}
    Z = \sum_{i \in S} \mathds{1}\left\{G_{i} \leq \frac{1-\delta_{1}}{2}\log\frac{n}{k} \right\};
\end{equation}
by the linearity of expectation, $\bE Z \leq k\left(\frac{k}{n}\right)^{\frac{1}{2} - o(1)}$, and by Markov's inequality, it follows that 
\begin{align}
    \bP\left(Z \geq \eta k\right) &\leq \frac{\bE Z}{\eta k} \nonumber \\
    &\leq \frac{ \left(\frac{k}{n}\right)^{\frac{1}{2}-o(1)}}{\eta} \nonumber \\
    \label{eq:markov_good_1} &\leq c\left(\frac{k}{n}\right)^{\frac{1}{2} - \lambda - o(1)} \\
    \label{eq:markov_good_2} &= o(1),
\end{align}
where $c > 0$ is a constant such that $\frac{1}{\eta} \leq c(\frac{n}{k})^{\lambda}$ for sufficiently large $n$ (since $\eta = \Omega((\frac{k}{n})^{\lambda})$), and \eqref{eq:markov_good_2} follows since $\lambda \in [0,\frac{1}{2})$ and $k=o(n)$. Hence, we have with probability $1 - o(1)$ that at least $k - \eta k$ defective items $i \in S$ have $G_{i} \geq \frac{1 - \delta_{1}}{2} \log \frac{n}{k}$ as claimed. 
\end{proof}

Lemma \ref{lem1} implies that the following holds with probability $1-o(1)$:
\begin{align}
    \exists &\tilde{S} \,:\, \big(\abs{\tilde{S}} = (1-\eta)k\big) \wedge \big(\tilde{S} \subseteq S\big) \nonumber \\
    &\wedge \bigg(\text{each $i \in \tilde{S}$ is in at least $\frac{1 - \delta_{1}}{2}\log\frac{n}{k}$ good tests}\bigg), \label{eq:Stilde_exists}
\end{align}
and this further implies
\begin{equation}
    \tilde{S} \in \cS_{\mathrm{alg}},
\end{equation}
which follows from the definition of $\cS_{\mathrm{alg}}$ in \eqref{eq:S_alg} along with $\dH(\tilde{S},S) = |S \setminus \tilde{S}| = \eta k$, $\dH(\hat{S}',S) \le 2\eta k$, and the triangle inequality.  The importance of $\tilde{S}$ is highlighted in the following lemma.

\begin{lemma} \label{lem:succ_conds}
    Conditioned on the event \eqref{eq:Stilde_exists}, and given $\tilde{S}$ specified in that event, the SUBSET algorithm succeeds provided that, for all $\ell = 1,2,\dotsc,5\eta k$, and every possible $\tilde{S}'$ formed by removing $\ell$ items from $\tilde{S}$ and replacing them with $\ell$ non-defective items, it holds that $\tilde{S}'$ explains fewer tests than $\tilde{S}$.  
    
    Moreover, when any set of $\ell$ items $\cD_\ell \subset \tilde{S}$ is removed from $\tilde{S}$, these removals collectively lead to at least $\frac{1-\delta_1}{2}\ell \log \frac{n}{k}$ tests no longer being explained by $\tilde{S} \setminus \cD_\ell$.
\end{lemma}
\begin{proof}
    In view of the decoding rule in Algorithm \ref{alg2}, and the fact that $\tilde{S} \in \mathcal{S}_{\rm alg}$, in order to prove the first statement, it suffices to show that the sets $\tilde{S}'$ formed in this manner (along with $\tilde{S}$ itself) collectively cover all of $\mathcal{S}_{\rm an}$ (and hence all of $\mathcal{S}_{\rm alg}$, since $\mathcal{S}_{\rm alg} \subseteq \mathcal{S}_{\rm an}$).  The range $\ell=1,2,\dotsc,5\eta k$ suffices because the Hamming distance upon performing $\ell$ swaps is $2\ell$, and any two sets in $\mathcal{S}_{\rm an}$ have Hamming distance at most $10\eta k$ due to the triangle inequality (see \eqref{eq:S_an}).  
    
     For the second claim, we note that by the final condition in \eqref{eq:Stilde_exists}, removing any item from $\tilde{S}$ leads to at least $\frac{1 - \delta_{1}}{2} \log \frac{n}{k}$ tests no longer being explained.  Thus, removing any set of $\ell$ items from $\tilde{S}$ leads to at least $\frac{1-\delta_1}{2}\ell \log \frac{n}{k}$ tests no longer being explained.
\end{proof}

We proceed to analyze the probability of some $\tilde{S}'$ from Lemma \ref{lem:succ_conds} explaining more tests than $\tilde{S}$, starting with the following lemma.

\begin{lemma} \label{lem:Stilde}
    Fix $\lambda \in (0, \frac{1}{2})$, let $\eta = \Theta((\frac{k}{n})^{\lambda})$, and suppose that $\theta > \frac{1}{3}$. Then:
    \begin{enumerate}[(i)]
        \item With probability $1-o(1)$, for any $\delta_{2} =  \omega(\frac{1}{\sqrt{T}})$, the number of positive tests is at most $\frac{T}{2}(1+\delta_{2})$.
        \item With probability $1-o(1)$, under the choice $\delta_{2} = \sqrt{\frac{16(3+ \frac{1}{\lambda})\eta k\log \frac{1}{\eta}}{T}}$, it holds for all $\bar{s} \in \cS_{\mathrm{an}}$ that the number of tests in which at least one item from $\bar{s}$ is present is at least $T(1 - 2^{-(1-\eta)(1+o(1))})(1 - \delta_{2})$.
    \end{enumerate}
\end{lemma}
\begin{proof}
    We begin with claim (i). Since each test is positive with probability $\frac{1}{2}$ and all tests are mutually independent, the number of positive tests is distributed as $\bin(T, \frac{1}{2})$. Thus, due to binomial concentration (Appendix \ref{app:bino}), it holds that there are at most $\frac{T}{2}(1+\delta_{2})$ positive tests with probability at least $1 - \exp(-\frac{\delta_2^2}{3}T)$, which is $1-o(1)$ when $\delta_{2} = \omega(\frac{1}{\sqrt{T}})$.
    
    We now proceed with claim (ii).  For a given test and $\bar{s} \in \cS_{\mathrm{an}}$, at least one item $i \in \bar{s}$ is included in it with probability $1 - \big(1 - \frac{\log 2 (1+o(1))}{k}\big)^{(1-\eta)k} = 1 - e^{-(1-\eta)\log 2(1+o(1))} = 1 - 2^{-(1 - \eta)(1+o(1))}$. For brevity, we say that $\bar{s}$ is \emph{present} in a test if at least one item from $\bar{s}$ is in the test. Thus, the number of tests in which $\bar{s}$ is present is distributed as $\bin(T,1- 2^{-(1-\eta)(1+o(1))})$. Denoting the mean of this distribution as $\mu_{\bar{s}}$ and using the multiplicative Chernoff bound (Appendix \ref{app:bino}), the number of tests in which $\bar{s} \in \cS_{\mathrm{an}}$ is present is at most $\mu_{\bar{s}} (1 - \delta_{2})$ with probability at most $\exp(-\frac{1}{2}\delta_{2}^{2}\mu_{\bar{s}}) \leq \exp(-\frac{1}{8}\delta_{2}^{2}T)$, where we have used the crude bound $\mu_{\bar{s}} \geq \frac{1}{4}T$ for sufficiently large $n$ (with $\eta = o(1)$). Substituting the specified value of $\delta_{2}$, this bound becomes
    \begin{align}
        \bP(\bin(T, 1- &2^{-(1-\eta)(1+o(1))}) \leq (1-\delta_{2})\mu_{\bar{s}}) \nonumber \\
        \label{eq:present_bound}&\leq \exp\left(-2\Big(3 + \frac{1}{\lambda}\Big)\eta k\log \frac{1}{\eta}\right).
    \end{align}
    Note that when $\theta > \frac{1}{3}$ and $\eta = \Theta((\frac{k}{n})^{\lambda})$ for some $\lambda \in (0, \frac{1}{2})$ as assumed in the lemma statement, the expression $\eta k \log\frac{1}{\eta}$ tends to $\infty$ as $n \to \infty$, meaning that the bound in \eqref{eq:present_bound} is $o(1)$. Consequently, we also have $\sqrt{\frac{16(3+\frac{1}{\lambda})\eta k\log\frac{1}{\eta}}{T}} = \omega(\frac{1}{\sqrt{T}})$, 
    so we can take this as the value of $\delta_{2}$ in (i).

We now use \eqref{eq:present_bound} along with a union bound to obtain a counterpart that holds \emph{simultaneously} for all $\bar{s} \in \cS_{\mathrm{an}}$. To do so, we first require an upper bound on the size of $\cS_{{\rm an}}$. Toward this, observe that we can calculate $\lvert \cS_{{\rm an}}\rvert$ directly via the following counting argument:
\begin{itemize}
    \item Consider $\bar{s} \in \cS_{{\rm an}}$ as being formed by taking $(1-\eta)k -i$ items from $S$ (recall that we implicitly condition on fixed $S = s$) and $i$ items from $S^{\mathsf{c}}$ for some $i \geq 0 $.
    \item Since $\lvert \bar{s}\rvert = (1-\eta)k$, the Hamming distance between $\bar{s}$ and $S$ satisfies $\dH(\bar{s}, S) = \eta k + 2i$. Moreover, since $\bar{s} \in \cS_{{\rm an}}$, it satisfies $\dH(\bar{s}, S) \leq 5\eta k$. Thus, it holds that $0 \leq i \leq 2\eta k$.
    \item For each $i$, the number of possible $\bar{s} \in \cS_{{\rm an}}$ that can be formed this way is ${k \choose (1-\eta)k - i}{n - k \choose i} ={k \choose \eta k +i}{n - k \choose i}$.
    \item Collectively this construction for $i = 0, \dots, 2\eta k$ produces all sets in $\cS_{{\rm an}}$, and it follows that
    \begin{equation}
        \label{eq:San_size} \lvert \cS_{{\rm an}}\rvert = \sum_{i=0}^{2\eta k} {k \choose\eta k + i}{n - k \choose i}.
    \end{equation}
\end{itemize}
We now proceed to bound $\lvert \cS_{{\rm an}}\rvert$:
\begin{align}
    \label{eq:San_size_bound_1} \lvert \cS_{{\rm an}}\rvert &= \sum_{i=0}^{2\eta k} {k \choose\eta k + i}{n - k \choose i} \\
    \label{eq:San_size_bound_2} &\leq (2\eta k + 1){k \choose 3 \eta k}{n \choose 2 \eta k} \\
    &= (2\eta k + 1) \exp\Big(3 \eta k \log\Big(\frac{1}{3 \eta }\Big)(1+o(1))  \nonumber \\
    \label{eq:San_size_bound_3} &\qquad \qquad + 2 \eta k \log \Big(\frac{n}{2 \eta k}\Big)(1+o(1))\Big) \\
    \label{eq:San_size_bound_4} &= \exp\left(\Big(5 + \frac{2}{\lambda}\Big) \eta k \log\Big(\frac{1}{\eta}\Big)(1+o(1))\right),
\end{align}
where \eqref{eq:San_size_bound_2} follows from bounding ${n - k \choose i}$ by ${n \choose i}$ and since the largest term in the sum is when $i = 2\eta k$, \eqref{eq:San_size_bound_3} holds since $\eta k = o(k)$ (i.e., $\eta = o(1)$), and \eqref{eq:San_size_bound_4} follows from writing $\frac{n}{k} = \Theta((\frac{1}{\eta})^{\frac{1}{\lambda}})$ and simplifying, as well as absorbing the prefactor $2 \eta k + 1$ into the $o(1)$ term.

With the bound on $\lvert \cS_{{\rm an}}\rvert$ in place, we are now in a position to take a union bound over all $\bar{s} \in \cS_{{\rm an}}$. The distribution of the number of tests in which $\bar{s}$ is present is the same for each $\bar{s}$ and satisfies the tail bound in \eqref{eq:present_bound}, and the number of these sets is upper bounded via \eqref{eq:San_size_bound_4}. Thus, the probability of any such set being present in less than $T(1 - 2^{-(1-\eta)(1+o(1))})(1 - \delta_{2})$ tests is upper bounded by
\begin{align}
    & \exp\left(\Big(5 + \frac{2}{\lambda}\Big) \eta k \log\Big(\frac{1}{\eta}\Big)(1+o(1)) \right) \nonumber \\
    &\qquad \qquad \qquad \times \exp\left(-2 \Big(3+ \frac{1}{\lambda}\Big)\eta k\log\frac{1}{\eta}\right)
    \\ 
    \label{eq:union_bound_3} &= \exp\left(-\eta k \log \Big(\frac{1}{\eta}\Big) (1-o(1))\right) \\
    \label{eq:union_bound_4} &= o(1),
\end{align}  
where \eqref{eq:union_bound_4} follows since $\eta k \log \frac{1}{\eta} \to \infty$ as established following \eqref{eq:present_bound}. Thus, all sets in $\cS_{\mathrm{an}}$ are present in at least $T(1 - 2^{-(1-\eta)(1+o(1))})(1 - \delta_{2})$ tests.
\end{proof}

Lemma \ref{lem:Stilde} hinges on the two specified assumptions on $\eta$ and $\theta$; we now proceed to argue that it is safe to assume both of these for the purpose of proving Theorem \ref{thm2}.  For $\eta$, this is because we only replaced the assumed scaling from $\Omega(\cdot)$ to $\Theta(\cdot)$ (and restricted $\lambda$ to $(0, \frac{1}{2})$ rather than $[0,\frac{1}{2})$); however, solving SUBSET with a smaller $\eta$ immediately implies the same for larger $\eta$, since one can simply remove an arbitrary subset of items from the final estimate.  For $\theta$, the restriction $\theta > \frac{1}{3}$ is valid because when $\theta \leq \frac{1}{3}$, a rate of $1$ is achievable even for exact recovery \cite{scarlett2016phase, coja2020optimal}, which implies solving SUBSET for arbitrary $\eta$.  Hence, for the rest of the proof, we adopt the stricter assumptions from Lemma \ref{lem:Stilde}, in particular noting that $\eta = o(1)$.


Using part (ii) of Lemma \ref{lem:Stilde} and the fact that $\tilde{S} \in \cS_{\mathrm{an}}$, we have that $\tilde{S}$ explains at least $T(1- 2^{-(1-\eta)(1+o(1))})$ tests with probability $1-o(1)$ (since $\tilde{S}$ explains every test it is present in due to $\tilde{S} \subseteq S$). Hence, and using part (i) of Lemma \ref{lem:Stilde}, the number of positive tests \emph{not} explained by $\tilde{S}$ 
is at most the following, with probability $1 - o(1)$:
\begin{align} 
    \frac{T}{2}(1 + \delta_{2}) - T(1 - 2^{-(1-\eta)(1+o(1))})(1 - \delta_{2}) \nonumber \\
   \label{eq:unexplained} = \frac{T}{2}\big(2^{\eta(1+o(1))} - 1 + \delta_{2}(3 - 2^{\eta(1+o(1))})\big).
\end{align}
It is also useful to note the following scaling on $\delta_2$ from Lemma \ref{lem:Stilde}:
\begin{align}
    \delta_{2} = \sqrt{\frac{16{(3+ \frac{1}{\lambda})}\eta k\log \frac{1}{\eta}}{T}} = \Theta\bigg( \sqrt{ \frac{\eta \log \frac{1}{\eta} }{ \log n }} \bigg) &= \Theta\big( \sqrt{\eta} \big) \nonumber \\
    &= n^{-\Omega(1)}, \label{eq:delta2_scaling}
\end{align}
where we applied $T = \Theta(k \log n)$ and $\log\frac{1}{\eta} = \Theta(\log n)$ (via $\eta = \Theta((\frac{k}{n})^{\lambda})$).

Using these findings and recalling Lemma \ref{lem:succ_conds}, the following lemma gives an upper bound on the probability that some $\tilde{S}'$ explains more tests than $\tilde{S}$ upon swapping out $\ell$ defectives for $\ell$ non-defectives.

\begin{lemma} \label{lem:fixed_ell}
    Suppose that the placements of defective items are such that the high-probability events \eqref{eq:Stilde_exists} and \eqref{eq:unexplained} hold, and consider $\tilde{S}$ from \eqref{eq:Stilde_exists}.  Fix some $\ell \in \{1, 2, \dots, 5\eta k\}$ and sets $\cD_\ell \subset \tilde{S}$ and $\cN_{\ell} \subset S^\mathsf{c}$ with $\abs{\cD_\ell} = \abs{\cN_{\ell}} = \ell$. Let $\tilde{S}' =  (\tilde{S} \setminus \cD_\ell) \cup \cN_{\ell}$ be the set formed by swapping out the $\ell$ defective items from $\cD_\ell$ with the $\ell$ non-defective items from $\cN_{\ell}$. Then, denoting the number of tests explained by $\tilde{S}$ (resp. $\tilde{S}'$) as $T_{\mathrm{exp}}(\tilde{S})$ (resp. $T_{\mathrm{exp}}(\tilde{S}')$), there exists a choice of $\delta_{1} > 0$ given in Lemma \ref{lem1} such that 
    \begin{align}
        \label{eq:fixed_ell_bound} \bP(T_{\mathrm{exp}}(\tilde{S}) &\leq T_{\exp}(\tilde{S}')) \leq \exp\left(-\Omega\left(\ell \sqrt{\log n} \cdot \log n\right)\right),
    \end{align}
    where the probability is over the randomness in the placements of non-defective items.
\end{lemma}
\begin{proof}
We first comment on the conditioning on \eqref{eq:Stilde_exists} and \eqref{eq:unexplained} (and implicitly, also on fixed $S=s$ as indicated earlier).  The reason \eqref{eq:Stilde_exists} only depends on the placements of defective items is that $\tilde{S} \subseteq S$ and non-defective placements do not affect whether a test is ``good'' for $i \in \tilde{S}$ (see Definition \ref{def5}).  The same reasoning applies to \eqref{eq:unexplained}, since this bound is on the ``number of positive tests not explained by $\tilde{S}$'', for which non-defectives are irrelevant.  Here we are proving a statement regarding the placements of non-defectives, and these are independent of the placements of defectives, so the preceding conditioning does not cause complications.

Our proof strategy consists of finding a ``worst case distribution'' on the number of positive tests the non-defective items in $\cN_{\ell}$ get placed in, and showing that the probability of $\tilde{S}'$ explaining more tests than $\tilde{S}$ under this distribution is small. The probability of at least one of the $\ell$ non-defective items in $\cN_{\ell}$ being included in any given test is $1 - \big(1 - \frac{\log 2 (1+o(1))}{k}\big)^{\ell} = \frac{\ell \log 2}{k}\big(1+o(1)\big)$, since $\ell = o(k)$. Thus, recalling the high probability upper bound on the number of tests left unexplained by $\tilde{S}$ given in \eqref{eq:unexplained}, we have that the number of tests that could be newly explained by $\cN_\ell$ is at most
\begin{equation}
    \label{eq:newly_explained} \frac{T}{2}\big(2^{\eta(1+o(1))} - 1 + \delta_{2}(3 - 2^{\eta(1+o(1))})\big) + \Psi_{\cD_\ell},
\end{equation}
where $\Psi_{\cD_\ell} \coloneqq T_{{\rm exp}}(\tilde{S}) - T_{{\rm exp}}( \tilde{S} \setminus \cD_\ell)$ is the number of previously-explained tests that are no longer explained after removing the items in $\cD_\ell$ from $\tilde{S}$. We require the inclusion of $\Psi_{\cD_\ell}$ in \eqref{eq:newly_explained} since it concerns tests no longer explained following the removal of items in $\cD_\ell$, and these tests may be ``re-explained'' by an item in $\cN_\ell$.  
Also note that since we are conditioned on \eqref{eq:Stilde_exists}, Lemma \ref{lem:succ_conds} implies that $\Psi_{\cD_\ell} \geq \frac{1-\delta_1}{2}\ell \log \frac{n}{k}$ for some $\delta_1 = 1 -o(1)$ (to be specified later), which will be useful later on in the proof.

From \eqref{eq:newly_explained}, it follows that the number tests that are newly explained by $\cN_{\ell}$ (including those ``re-explained'') is stochastically dominated\footnote{For two random variables $X_1, X_2$, we say that $X_1$ is \emph{stochastically dominated} by $X_2$ (or $X_2$ \emph{stochastically dominates} $X_1$) if $\bP(X_1 \geq x) \leq \bP(X_2 \geq x),\ \forall x \in \bR$.  This means that $X_2$ is ``greater than or equal to $X_1$'' in a probabilistic sense. 
} by the following:\footnote{
    Our analysis focuses on the positive tests that these $\ell$ non-defectives get placed in.  Any negative tests they get placed in will only decrease the number of explained tests further, since by Definition \ref{def:explained} any non-defective in a negative test does not explain any tests.  However, for our purposes it suffices to use a naive upper bound that replaces this decrease by zero.
    }
\begin{align}
    E \sim \bin\bigg(\frac{T}{2}\big(2^{\eta(1+o(1))} - 1 + \delta_{2}(3 - 2^{\eta(1+o(1))})\big) \nonumber \\
    + \Psi_{\cD_\ell}, \frac{\ell \log 2}{k}(1+o(1))\bigg). \label{eq:E_distr}
\end{align}
Moreover, the definitions of $E$ and $\Psi_{\cD_\ell}$ directly imply that
\begin{equation}
    \label{eq:event_inclusion} \{T_{\exp}(\tilde{S}) \leq T_{\exp}(\tilde{S}')\} = \{E \geq \Psi_{\cD_\ell}\},
\end{equation}
since when forming $\tilde{S}'$ from $\tilde{S}$ there is a ``loss'' of $\Psi_{\cD_\ell}$ from removing $\cD_\ell$, but then a ``gain'' of $E$ from adding $\mathcal{N}_\ell$.  

Thus, it suffices to upper bound $\bP(E \geq \Psi_{\cD_\ell})$.  To do so, we first analyze the mean:
\begin{align}
    \mu_{E} &\coloneqq \bE[E] \nonumber \\
    &= \Bigg(\frac{T\big(2^{\eta(1+o(1))} - 1 + \delta_{2}(3 - 2^{\eta(1+o(1))})\big)\ell \log 2}{2k} \nonumber \\
    & \qquad \qquad \qquad \qquad+ \frac{\ell \Psi_{\cD_\ell}\log 2}{k}\Bigg)\big(1+o(1)\big). \label{eq:E_mean}
\end{align}
We also note the following asymptotic simplification for $2^{\eta(1+o(1))}-1$ when $\eta = o(1)$:
\begin{align}
    2^{\eta(1+o(1))} &= e^{\eta\log 2(1+o(1))} = 1 + \eta \log 2\big(1 +o(1)\big) \nonumber \\ 
    \label{eq:17} &\implies  2^{\eta(1+o(1))} - 1 = \eta\log 2\big(1+o(1)\big).
\end{align}
As a result, the term $2^{\eta} -1 = O(\eta)$ is dominated by the expression $\delta_{2}(3 - 2^{\eta}) = \Theta(\delta_2) = \Theta\big(\sqrt{\eta})$ (see \eqref{eq:delta2_scaling}). Thus, we have the following asymptotic expression for the mean $\mu_{E}$:
\begin{align} 
    \mu_{E} &= \bigg(\frac{T\delta_{2}(3 - 2^{\eta(1+o(1))})\ell \log 2}{2k} \nonumber \\
    \label{eq:18} &\qquad \qquad \qquad \qquad \quad + \frac{\ell \Psi_{\cD_\ell}\log 2}{k}\bigg)\big(1+o(1)\big) \\
    &= \bigg(\frac{\left(k \log_{2}\frac{n}{k}\right)(1+\epsilon)\ell\delta_{2}(3 - 2^{\eta(1+o(1))})\log 2}{2k} \nonumber \\
    \label{eq:19} &\qquad \qquad \qquad \qquad \quad + \frac{\ell \Psi_{\cD_\ell}\log 2}{k}\bigg)\big(1+o(1)\big) \\
    &= \bigg(\frac{(\log\frac{n}{k})(1 + \epsilon)\ell \delta_{2}(3 - 2^{\eta(1+o(1))})}{2} \nonumber \\
    \label{eq:20} &\qquad \qquad \qquad \qquad \quad + \frac{\ell \Psi_{\cD_\ell}\log 2}{k}\bigg)\big(1+o(1)\big),
\end{align}
where \eqref{eq:19} follows by substituting $T = \big(k\log_{2}\frac{n}{k}\big)(1+\epsilon)$.
Hence, the scaling of $\mu_{E}$ is
\begin{equation}
    \label{eq:mu_e_scaling} \mu_E = \Theta\left(\ell\Big(\delta_{2}\log n + \frac{\Psi_{\cD_\ell}}{k}\Big)\right).
\end{equation}
We can compare this to the above-established threshold of $\Psi_{\cD_\ell}$. Defining $\delta_3 > 0$ to satisfy $\Psi_{\cD_\ell} = (1+\delta_3)\mu_E$, we have via \eqref{eq:mu_e_scaling} that
\begin{align}
    \label{eq:delta3_rearrange_1} \Psi_{\cD_\ell} &= (1+\delta_3) \mu_E\\
    \label{eq:delta3_rearrange_2} &= \Theta\left((1+\delta_3)\ell \Big(\delta_2 \log n  + \frac{\Psi_{\cD_\ell}}{k}\Big)\right) \\
    \label{eq:delta3_rearrange_3} &\implies 1 + \delta_3 = \Theta\left(\frac{\Psi_{\cD_\ell}}{\ell (\delta_2 \log n + \frac{\Psi_{\cD_\ell}}{k})}\right).
\end{align}
Since the quantity in the right hand side of \eqref{eq:delta3_rearrange_3} is increasing in $\Psi_{\cD_\ell}$, and recalling that $\Psi_{\cD_\ell} \geq \frac{1-\delta_1}{2}\ell \log \frac{n}{k}$ when conditioned on \eqref{eq:Stilde_exists}, we can therefore deduce that
\begin{align}
    \label{eq:delta3_rearrange_4} 1 + \delta_3 &= \Omega\left(\frac{\frac{1-\delta_1}{2}\ell \log \frac{n}{k}}{\ell(\delta_2 \log n + \frac{(1-\delta_1)\ell \log \frac{n}{k}}{2k})}\right) \\
    \label{eq:delta3_rearrange_5} &= \Omega\left(\frac{1}{\frac{\delta_2}{1-\delta_1} + \frac{\ell}{k}}\right) \\ 
    \label{eq:delta3_rearrange_6} &= \Omega\left(\frac{1}{\frac{\delta_2}{1-\delta_1} + \eta}\right)\\
    \label{eq:delta3_rearrange_7} &= \Omega\left(\frac{1-\delta_1}{\delta_2}\right),
\end{align}
where \eqref{eq:delta3_rearrange_5} is obtained by dividing through by $\frac{1-\delta_1}{2}\ell \log \frac{n}{k}$ and simplifying, \eqref{eq:delta3_rearrange_6} uses $\ell = O(\eta k)$ and the fact that the fraction is decreasing in $\ell$, and \eqref{eq:delta3_rearrange_7} follows since $\delta_2 = \Theta(\sqrt{\eta}) = \omega(\eta)$ (see \eqref{eq:delta2_scaling}). Choosing $\delta_{1} = 1 - \frac{1}{\sqrt{\log n}}$ and using the scaling of $\delta_2$ in \eqref{eq:delta2_scaling}, the scaling in \eqref{eq:delta3_rearrange_7} becomes $1+\delta_3 = \Omega(\frac{1}{\delta_2\sqrt{\log n}}) = \frac{n^{\Omega(1)}}{\sqrt{\log n}} = \omega(1)$, which implies the same scaling for $\delta_3$. Thus, we see that $\mu_{E}$ grows much more slowly than $\Psi_{\cD_\ell}$. As such, we can use binomial concentration bounds to show that $\bP(E \geq \Psi_{\cD_\ell})$ is small. 

To be more specific, we use the aforementioned scaling of $\delta_3$, the fact that $\Psi_{\cD_\ell} = (1+\delta_3)\mu_E$, and the multiplicative Chernoff bound (Appendix \ref{app:bino}), to show that
\begin{align}
    \label{eq:21} \bP\left(E \geq \Psi_{\cD_\ell}\right) &\hspace{-0.75pt}= \bP(E \geq (1 + \delta_{3})\mu_{E}) \\
    \label{eq:22} &\hspace{-0.75pt}\leq \exp\left(-\big((1 +\delta_{3})\log(1+\delta_{3}) - \delta_{3}\big)\mu_{E}\right) \\
    \label{eq:23} &\hspace{-0.75pt}= \exp(-(1 +\delta_{3})\log(1+\delta_{3})\mu_{E}(1-o(1))) \\
    \label{eq:24} &\hspace{-0.75pt}= \exp\left(-\Omega\Big(\ell\sqrt{\log n} \cdot \log\delta_{3}\Big)\right) \\
    \label{eq:e_bound} &\hspace{-0.75pt}=\exp\left(-\Omega\left(\ell \sqrt{\log n} \cdot \log n\right)\right),
\end{align}
where \eqref{eq:23} follows since $\delta_{3} = \omega(1)$, \eqref{eq:24} follows since $(1+\delta_{3})\mu_{E} = \Psi_{\cD_\ell} \geq \frac{1-\delta_{1}}{2}\ell \log n = \Omega(\ell \sqrt{\log n})$, and \eqref{eq:e_bound} follows since $\delta_{3} =\frac{n^{\Omega(1)}}{\sqrt{\log n}}$ and hence $\log \delta_3 = \Omega(\log n)$. 
Lemma \ref{lem:fixed_ell} then follows by recalling that $\bP(T_{\exp}(\tilde{S})\leq T_{\exp}(\tilde{S}')) = \bP(E \geq \Psi_{\cD_\ell})$ due to \eqref{eq:event_inclusion}. 
\end{proof}

With this result established, we are now in a position to show that a union bound over all possible $(\cD_\ell, \cN_{\ell})$, followed by a second union bound over all the values of $\ell$, still leads to $o(1)$ probability of $\tilde{S}$ not explaining the most tests.  This is formalized in the following lemma.

\begin{lemma} \label{lem:union_ell}
    Under the setup and notation of Lemma \ref{lem:fixed_ell},\footnote{In particular, we are again conditioning on \eqref{eq:Stilde_exists} and \eqref{eq:unexplained}, and the probability is over the randomness in the placements of non-defective items.} for each $\ell \in\{1,2,\dots, 5\eta k\}$, define 
    \begin{equation}
        \label{eq:error_ell} \cE_{\ell} \coloneqq \bigcup_{(\cD_\ell, \cN_\ell)}\left\{T_{\exp}(\tilde{S}) \leq T_{\exp}((\tilde{S} \setminus \cD_\ell) \cup \cN_{\ell})\right\},
    \end{equation}
    where each $(\cD_\ell, \cN_\ell)$ pair implicitly satisfies $\lvert \cD_\ell\rvert = \lvert \cN_\ell \rvert = \ell$, as the event that at least one way of swapping out $\ell$ items from $\tilde{S}$ with $\ell$ non-defective items leads to more tests being explained.  Then, we have the following as $n \to \infty$:
    \begin{equation} 
        \bP\left(\bigcup_{\ell = 1}^{5\eta k} \cE_{\ell}\right) \to 0. \label{eq:union_to_zero}
    \end{equation}
\end{lemma}
\begin{proof}
We begin by deriving a bound on $\bP(\cE_{\ell})$ for fixed $\ell$.  For the event corresponding to any single choice of $(\cD_\ell, \cN_{\ell})$ in \eqref{eq:error_ell}, the associated probability is upper bounded via Lemma \ref{lem:fixed_ell}.
A simple counting argument reveals that there are ${(1-\eta)k \choose \ell}{n-k \choose\ell}$ such pairs, which can be further upper bounded by ${k \choose \ell}{n \choose \ell}$. Hence, the union bound gives
\begin{align}
   \bP(\cE_\ell)\label{eq:pe_ell_3} &\leq {k \choose \ell}{n\choose \ell}\exp\left(-\Omega\big(\ell \sqrt{\log n} \cdot \log n\big)\right) \\
   \label{eq:pe_ell_5} &\leq \exp\left(\ell \log k  + \ell \log n -\Omega\big(\ell \sqrt{\log n} \cdot \log n\big)\right) \\
   \label{eq:pe_ell_6} &= \exp\left(-\Omega\big(\ell \sqrt{\log n} \cdot \log n\big)\right),
\end{align}
where \eqref{eq:pe_ell_5} follows since ${n \choose k} \leq n^k$, and \eqref{eq:pe_ell_6} follows since $\ell \log k \le \ell \log n = o(\ell \sqrt{\log n} \cdot \log n)$. 

We can then take another union bound over the values of $\ell$ to obtain
\begin{align}
    \label{eq:pe_2} \bP\left(\bigcup_{\ell = 1}^{5\eta k}\cE_{\ell}\right) &\leq \sum_{\ell=1}^{5\eta k} \bP(\cE_{\ell}) \\
    \label{eq:pe_3} &\leq \sum_{\ell = 1}^{5\eta k} \exp\left(-\Omega\big(\ell \sqrt{\log n} \cdot \log n\big)\right) \\
    \label{eq:pe_4} &= o(1),
\end{align}
since the number of $\ell$ values is upper bounded by $k = e^{O(\log n)}$.  This gives the desired result \eqref{eq:union_to_zero}.

\end{proof}

Finally, the proof of Theorem \ref{thm2} is completed by simply combining Lemma \ref{lem:union_ell} with Lemma \ref{lem:succ_conds}.


\subsection{Proof of Corollary \ref{corollary1}} \label{sec:cor_proof}

We first recall the result in \cite[Theorem 1]{truong2020all}, which is stated as follows. 
\begin{lemma} \label{lem7} \emph{\cite[Theorem 1]{truong2020all}}
    Fix $\alpha \in (0,1)$ and $\epsilon > 0$, and suppose that $T \geq (1- \alpha)\big(k\log_{2}\frac{n}{k}\big)(1 + \epsilon)$. Then there exists a group testing algorithm returning an estimate $\hat{S}'$ such that $\max\{\abs{S \setminus \hat{S}'}, \abs{\hat{S}' \setminus S}\} \leq \alpha k$ with probability $1 - o(1)$.
\end{lemma}
This proof of this result is based on a simple argument that we outline as follows:
\begin{itemize}
    \item A fraction $\alpha - \xi$ of items, for arbitrarily small $\xi > 0$, is ``deleted'' and marked as nondefective, without performing any tests on them.
    \item Among the remaining items (of which there are $(1-\alpha+\xi)n$, with roughly $(1-\alpha+\xi)k$ being defective), apply the approximate recovery algorithm with from \cite{scarlett2016phase} with a sufficiently small distortion parameter.  (With a slight modification due to the number of defectives only being known approximately rather than exactly.)
\end{itemize}
A simple analysis then shows that there are $(\alpha - \xi)k(1+o(1))$ false negatives in the first step due to Hoeffding's inequality for the Hypergeometric distribution \cite{hoeffding1963probability}, and at most $\beta k$ false negatives and false positives in the second step (for arbitrarily small $\beta > 0$) by the guarantees in \cite{scarlett2016phase}.  Note that the $1-\alpha$ factor arises in the number of tests due to having roughly $(1-\alpha)k$ non-deleted defectives.

We can use the same idea for the SUBSET problem; since the reasoning is all analogous to that of \cite{truong2020all}, we omit some details and focus mainly on the differences.  We delete a fraction of items arbitrarily close to $\alpha$ and mark them as non-defective. Among those remaining, we apply the approximate recovery algorithm from \cite{scarlett2016phase} with $\beta = \eta$ for some $\eta = o(1)$ to obtain $\hat{S}'$. As noted in \cite{truong2020all}, a result in \cite[Appendix B]{scarlett2018noisy} shows that this algorithm still succeeds with high probability even when the number of defectives is only known within a factor of $1\pm o(1)$. We then apply the SUBSET algorithm (Algorithm \ref{alg2}) using the reduced ground set and the two-sided estimate $\hat{S}'$.  While we are unable to directly apply Algorithm \ref{alg2} and Theorem \ref{thm2}, since these assume that the number of defectives is known, we can easily modify the algorithm and proof to show the result still holds when it is only known to within a factor of $1 \pm o(1)$. Specifically, letting $K'$ be the (random) number of non-deleted defectives and defining $k' \coloneqq \bE K' = (1-\alpha+\xi)k$, we observe the following: 
\begin{itemize}
    \item A simple concentration argument reveals that $\underline{k}' \le K' \le \overline{k}'$ with probability $1-o(1)$, where $\underline{k}' = k'(1-o(1))$ and $\overline{k}' = k'(1+o(1))$.
    \item We cannot set the Hamming distance condition in Algorithm \ref{alg2} to $\dH(\hat{S}', \bar{S}) \leq 3\eta k'$, but we can safely replace $K'$ by the upper bound $\overline{k}'$.
    \item Similarly, the output set size cannot be set to $(1-\eta)K'$, but it suffices to use $|\bar{S}| = (1-\eta)\underline{k}'$.
    \item Due to the above modifications, in the definition of $\cS_{\rm an}$ in \eqref{eq:S_an} the constant $5$ becomes $5+o(1)$, but this substitution is inconsequential throughout the analysis (in fact, even increasing the constant, say from $5$ to $6$, would be inconsequential).
\end{itemize}
Since $\eta = o(1)$ and $\underline{k}' = k'(1-o(1))$, the resulting output is a set of size $|\hat{S}| = (1-\alpha+\xi)k(1-o(1))$, which is at least $(1-\alpha)k$ when $n$ is large enough.  The same $1-\alpha+\xi$ factor then appears in the number of tests, which gives the desired result since $\xi$ is arbitrarily small. \qed

\subsection{Proof of Theorem \ref{thm3} (Converse for SUPERSET)}
We now turn to the proof of Theorem 3, which centers around a commonly-used notion of \emph{masking}, defined as follows.

\begin{definition} \label{def:masked}
    An item \(i \in [n]\) is \emph{masked} if for every \(t \in \{1,2,\dots, T\}\) such that \(\sfX_{ti} =1\) there exists \(j \in S\) with $j \ne i$ such that \(\sfX_{tj} = 1\).  That is, every test that includes $i$ also includes at least one other defective.
\end{definition}

Intuitively, an item being masked means that it does not affect the test outcomes of any tests it is placed in, meaning the tests provide no information about its defectivity status. Thus, the best any decoder can do is guess the more likely of the two possibilities. We use ideas from \cite{coja2020optimal, bay2022optimal}, showing that in the exact recovery case, for $\theta$ close to 1, there exists a large number of masked items with high probability, and the existence of such masked items leads to a high probability of failure. 
They then transfer this ``high $\theta$'' result to any $\theta \in \big(\frac{\log 2}{1 + \log 2}, 1\big)$ by an argument based on ``dummy non-defective items'', which we will also use. (Note that for $\theta \in \big(0,\frac{\log 2}{1+\log 2}\big]$ the rate in Theorem \ref{thm3} is 1, and this converse is well-known and holds even for two-sided approximate recovery  \cite{scarlett2016phase,truong2020all}).

We use these ideas to find similar high probability bounds on the number of masked items, suitably modified to suit our approximate recovery setting, for rates above $R^{*}$ (defined in \eqref{eq:10}). In particular, we show that there are $\omega(1)$ masked defective items and $\omega(k)$ masked non-defective items with high probability, regardless of test design. Intuitively, this will lead to a converse for SUPERSET because when attempting to form an estimate $\hat{S}$ of $S$ of size $(1 + \eta)k$ such that $\hat{S} \supseteq S$, the decoder will be unable to distinguish masked defectives from masked non-defectives.  Since there are many more of the latter, any decoder will have a low probability of choosing all masked defectives in its estimate.

We now formalize the above intuition. To do so, we first show that for rates above $R^{*}$ there are $\omega(k)$ masked non-defective items with high probability. We split this into two cases, namely, $\theta \in \big(\frac{\log 2}{1 + \log 2}, \frac{1}{2}\big)$ and $\theta \in \big[\frac{1}{2}, 1\big)$.

\medskip
\noindent \textbf{\underline{Case 1:}} $\theta \in \big(\frac{\log 2}{1 + \log 2}, \frac{1}{2}\big)$ 
\medskip

\noindent For convenience, we will first establish the claim of $\omega(k)$ masked non-defectives under the assumption that every item appears in at most $\log^3n$ tests.  Intuitively, this is a mild condition because good designs are known to place each item in $O(\log n)$ tests (e.g., see \cite{coja2020optimal}).  Nevertheless, it does not cover all possible tests designs, and accordingly, we will drop this assumption and handle \emph{arbitrary designs} at the end of the analysis of this case.

We use the following result adapted from \cite{coja2020optimal}, which we then use to prove our claims for the SUPERSET problem.

\begin{lemma} \label{lem4}\emph{(Adapted from \cite[Propositions 3.4 and 3.5]{coja2020optimal})}
    \ Let $M_{1}$ denote the number of masked defective items, where there are $k = \Theta(n^{\theta})$ defective items present within a population of size $n$. Let $\xi > 0$ be arbitrarily small, and fix $\epsilon > 0$ and $\theta'$ sufficiently close to $1$ such that $2(1-\theta') < \xi < \epsilon \theta'$. Then if $\theta \in \big(\frac{\log 2}{1+\log 2}, \frac{1}{2}\big)$ and $T \leq \big(\frac{1}{\log 2}k \log_{2}k\big)(1 - \epsilon)$, it holds that
    \begin{equation} \label{eq:30}
        \bP\left(M_{1}\geq \frac{1}{4}k^{\epsilon - \frac{\xi}{\theta'}}\right) = 1 - o(1).
    \end{equation}
\end{lemma}

In Appendix \ref{sec:adaptations}, we describe how we obtain this result from the results in \cite{coja2020optimal}.

We now proceed to show that Lemma \ref{lem4} implies that there are $\omega(\frac{n}{k})$ (and hence $\omega(k)$ since $\theta < \frac{1}{2}$) masked non-defective items with high probability when $T \leq \big(\frac{1}{\log 2}k \log_{2}k\big)(1 - \epsilon)$. To do so, we consider $S$ being generated in the following manner:
\begin{enumerate}
    \item Generate $S' \subseteq [n]$ by including each item in it independently with probability $q' = \frac{k + \sqrt{k}\log n}{n}$.
    \item Remove $\max\{\abs{S'}-k, 0\}$ items uniformly at random from $S'$ to form $S$.
\end{enumerate}
Defining the event $\mathscr{B} = \{k \leq \abs{S'} \leq k + 2\sqrt{k}\log n\}$, it follows from the multiplicative Chernoff bound (see Appendix \ref{app:bino}) that $\bP(\mathscr{B}) = 1-o(1)$. Moreover, conditioned on $\abs{S'} \geq k$ (which is implied by $\mathscr{B}$), the resulting distribution of $S$ is indeed the combinatorial prior due to the symmetry of the construction. Thus, an event has $o(1)$ probability under the above distribution if and only if it has $o(1)$ probability under the combinatorial prior.  
Additionally, letting $M_1'$ be the number of masked defectives after the first step of the construction (i.e., the number of defectives that would be masked if $S'$ were the defective set) and $M_1$ denote the number after the second step, it holds that $M_{1}' \geq M_{1}$.  This is because converting defective items to non-defective in step 2 can only decrease the number of masked items (and, in turn, masked defectives), since items that were masked by such ``converted'' items may no longer be masked. 

Using $M'_1 \ge M_1$ and Lemma \ref{lem4}, setting $\gamma \coloneqq \epsilon - \frac{\xi}{\theta'} > 0$, and denoting the total number of masked items after the first step of the construction by $M'$ (i.e., the analog of $M_1'$ but with both defectives and non-defectives), we have
\begin{align}
   \label{eq:31} &1 - o(1) = \bP\left(M_{1}' \geq \frac{1}{4}k^{\gamma}\right) \\
   \label{eq:32} &= \sum_{m' = \frac{1}{4}k^{\gamma}}^{n} \bP(M' = m') \bP\left(M_{1}' \geq \frac{1}{4}k^{\gamma} \mid M' = m'\right) \\
   &= \underbrace{\sum_{m' = \frac{1}{4}k^{\gamma}}^{\frac{n}{k}k^{\frac{\gamma}{2}}} \bP(M' = m') \bP\left(M'_{1} \geq \frac{1}{4}k^{\gamma} \mid M' = m'\right)}_{\coloneqq \cA_{1}}  \nonumber \\
   \label{eq:33} &+ \underbrace{\sum_{m' = \frac{n}{k} k^{\frac{\gamma}{2}} + 1}^{n} \bP(M' = m') \bP\left(M_{1}' \geq \frac{1}{4}k^{\gamma} \mid M' = m'\right)}_{\coloneqq \cA_{2}}.
\end{align}
We proceed by showing that $\cA_{1} = o(1)$, which in turn implies $\cA_{2} = 1 - o(1)$. 
Under the i.i.d.~prior, it is known that the posterior probability of being defective for each masked item matches its prior \cite[Lemma 3.1]{bay2022optimal}, which can be shown via Bayes' rule and the fact that a single masked item's defectivity status has no impact on any test outcomes. Since the first step of the generation of $S$ is i.i.d.~with parameter $q'$, it holds that any masked item following this first step has posterior probability $q'$ of being defective. As a result, this implies that $\mu'_1(m') \coloneqq \bE[M_{1}' \mid M' = m'] = m'\frac{k + \sqrt{k}\log n}{n} \leq k^{\frac{\gamma}{2}}\big(1+o(1)\big)$ when $m' \leq \frac{n}{k}k^{\frac{\gamma}{2}}$. Thus, Markov's inequality implies that
\begin{align}
   \label{eq:36} \cA_{1}  &\leq \sum_{m' = \frac{1}{4}k^{\gamma}}^{\frac{n}{k}k^{\frac{\gamma}{2}}} \bP(M' = m') \frac{\mu_1'(m')}{\frac{1}{4}k^\gamma} \\
   \label{eq:38} &\leq \sum_{m' = \frac{1}{4}k^\gamma}^{\frac{n}{k}k^{\frac{\gamma}{2}}} \bP(M' = m') \frac{4k^{\frac{\gamma}{2}}(1+o(1))}{k^\gamma} \\
    \label{eq:39} &\leq 4k^{-\frac{\gamma}{2}}(1+o(1)) \\
   \label{eq:40} &= o(1),
\end{align}
where 
\eqref{eq:38} follows from the above bound on $\mu_1'(m')$ for $m' \leq \frac{n}{k}k^\frac{\gamma}{2}$, and \eqref{eq:39} follows by bounding the sum of the probabilities by 1. This implies that $\cA_{2} = 1 - o(1)$, which in turn implies via \eqref{eq:33} that
\begin{align}
   \label{eq:41} &1 - o(1) = \cA_{2} \\
   \label{eq:42} &= \sum_{m'=\frac{n}{k}k^{\frac{\gamma}{2}}+1}^{n} \bP(M' = m') \bP\left(M_{1}' \geq \frac{1}{4}k^{\gamma} \mid M' = m'\right) \\
   \label{eq:43} &\leq \sum_{m'=\frac{n}{k}k^{\frac{\gamma}{2}} +1}^{n} \bP(M' = m') \\
   \label{eq:44} &= \bP\left(M' \geq \frac{n}{k}k^{\frac{\gamma}{2}}\right).
\end{align}
Hence, $M' = \omega(\frac{n}{k})$ with probability approaching one.  We proceed to show that the same is true for $M$, the number of masked items after the second step.  The idea is to exploit the fact that $S'$ and $S$ are sufficiently ``close together'' under the high probability event $\mathscr{B}$.

In more detail, let $\cM'$ and $\cM$ be the sets of masked items under $S'$ and $S$ respectively (yielding $\abs{\cM'} = M'$ and $\abs{\cM} = M$).  We consider the probability for a given $i \in \cM'$ that $i \notin \cM$ (i.e., the probability that a specific masked item loses its masked status between the first and second steps).  Implicitly conditioning on the first step having been done (with $\mathscr{B}$ holding), letting $\bP_2(\cdot)$ denote probability with respect to the randomness in the second step, and letting $\cM_t'(i) \subseteq S'$ denote the set of items masking a given item $i \in \cM'$ in a given test $t$, we have for each $i \in \cM'$ that
\begin{align}
    \label{eq:unmasked_1} \bP_2(i \notin \cM) &= \bP_2\left(\bigcup_{t: \sfX_{ti}=1}\left\{j \notin S, \,\forall j \in \cM_t'(i)\right\}\right) \\
    \label{eq:unmasked_2} &\leq\sum_{t:\sfX_{ti}=1} \bP_2\left(j \notin S, \, \forall j \in \cM_t'(i)\right) \\
    \label{eq:unmasked_3} &\leq \sum_{t: \sfX_{ti}=1} \frac{2\sqrt{k}\log n}{k} \\
    \label{eq:unmasked_4} &\leq \log^3n \frac{2\log n}{\sqrt{k}},
\end{align}
where:
\begin{itemize}
    \item \eqref{eq:unmasked_1} follows since for an item to no longer be masked after the second step, there needs to be at least one test such that all the defective items masking it are converted to non-defective;
    \item \eqref{eq:unmasked_2} follows from the union bound;
    \item \eqref{eq:unmasked_3} follows by bounding the probability that \emph{all} $j \in \cM_t'(i)$ are converted to non-defective by the probability of this happening for a \emph{single} $j$ (e.g., the one with the lowest index). We can then bound the probability of this single item being converted by $\frac{2\sqrt{k}\log n}{k}$, since under the event $\mathscr{B}$ at most $2\sqrt{k}\log n$ such items get chosen uniformly at random from a set of size at least $k$ to be made non-defective;
    \item \eqref{eq:unmasked_4} follows from the assumption that each item is in at most $\log^3n$ tests, as introduced at the start of this case (and to be dropped shortly).
\end{itemize}
Letting $\Delta = M' - M$ denote the number of items that lose their masked status after the second step,
it follows that $\bE \Delta \leq \frac{2\log^4n}{\sqrt{k}}M'$. Markov's inequality then gives $\bP_2(\Delta \geq \frac{\log^5n}{\sqrt{k}}M') = o(1)$, which implies $M \geq M'(1-o(1))$ with probability $1-o(1)$. Combining this with the fact that $M' \geq \frac{n}{k}k^{\frac{\gamma}{2}}$ with probability $1-o(1)$,
we can thus infer there are at least $ \frac{n}{k}k^{\frac{\gamma}{2}}(1-o(1)) - k \geq \frac{n}{2k}k^{\frac{\gamma}{2}}  = \omega(\frac{n}{k})$ masked non-defective items with probability $1 - o(1)$. Since $\theta < \frac{1}{2}$, this quantity is $\omega(k)$.

{\bf Dropping the bounded tests-per-item assumption.} 
We now argue that the same finding is true for \emph{any} test design. To do so, we require the following results adapted from \cite{coja2020optimal, bay2022optimal}.
\begin{lemma}
    \label{lem:degrees} \emph{(Adapted from \cite[Lemma 3.3]{bay2022optimal} and \cite[Lemma 3.8]{coja2020optimal})} Let $\Gamma = \frac{n}{k} \log n$, and let $\cI_t \subseteq [n]$ be the set of items included in test $t$. Then, if $T = O(k\log n)$\emph{:}
    \begin{itemize}
        \item[(i)] It holds for any test design $\sfX$ that
        \begin{equation}
            \label{eq:test_degree} \hspace{-3pt} \bP(\exists t \in \{1,2,\dots, T\} : \lvert \cI_t \rvert >\Gamma \wedge \lvert \cI_t \cap S \rvert =0) = o(1).
        \end{equation}
        That is, it holds with probability $1-o(1)$ that every test containing more than $\Gamma$ items is positive.
        \item[(ii)] If all tests in $\sfX$ contain at most $\Gamma$ items, then there are $n-o(n)$ items that appear in at most $\log^3 n$ tests.
    \end{itemize}
\end{lemma}

We provide the proof of this lemma in Appendix \ref{sec:adaptations}. Now, consider an arbitrary test design $\sfX$, and partition the items $[n]$ into $(\cN_\mathrm{low}, \cN_\mathrm{high})$, where $\cN_\mathrm{low}$ is the set of items that appear in at most $\log^3n$ tests, and $\cN_\mathrm{high}$ contains the remaining items.  By the first part of Lemma \ref{lem:degrees}, we may assume without loss of generality that no test in $\sfX$ contains more than $\frac{n}{k}\log n$ items, since with probability $1-o(1)$, all tests containing more items are positive (and thus their outcomes could have been guessed reliably without performing them). Under this assumption, the second part of Lemma \ref{lem:degrees} implies that $\abs{\cN_\mathrm{low}} = n-o(n)$.

The key observation is that while Lemma \ref{lem4} is stated for the number of masked defectives, the proof in \cite{coja2020optimal} (under the condition $\abs{\cN_\mathrm{low}} = n-o(n)$) actually shows this result for \emph{masked defectives in $\cN_\mathrm{low}$}; this can directly be seen from the first step after Lemma 3.8 in \cite{coja2020optimal}. In other words, Lemma \ref{lem4} still remains true when we replace $M_1$ by $M_1^{\rm low}$, defined as the number of masked defectives in $\cN_\mathrm{low}$.  We may then apply steps \eqref{eq:31}-\eqref{eq:unmasked_4} with all such quantities (i.e., $M_1$, $M'_1$, $M'$, $M$, and the associated sets $\mathcal{M},\mathcal{M}'$) replaced by their counterparts that only count items in $\cN_\mathrm{low}$; apart from this modification, the analysis remains identical.  It follows that there are at least $\frac{n}{2k}k^{\frac{\gamma}{2}}$ masked non-defectives in $\cN_{\mathrm{low}}$ with probability $1-o(1)$, which trivially implies the same lower bound on the \emph{total} number of masked non-defectives.


\medskip
\noindent\textbf{\underline{Case 2:}} $\theta \in [\frac{1}{2}, 1)$
\medskip

\noindent The steps in this case are similar to Case 1, so we only provide an outline while highlighting the key differences.  Analogous to Case 1, we initially assume that every item is included in at most $n^{\xi}$ tests for some arbitrarily small $\xi > 0$;\footnote{This is only different from $\log^3 n$ because we use a result from \cite{bay2022optimal} instead of \cite{coja2020optimal}, and they happened to choose slightly different values.  The analysis in \cite{bay2022optimal} could easily be modified to use $\log^3 n$ instead, but we find it more convenient to apply their results exactly as stated.} this assumption will later be dropped in the same way as above.

To show there are $\omega(k)$ masked non-defectives with high probability, instead of using the result from \cite{coja2020optimal}, in this case we use the following result adapted from \cite{bay2022optimal}. 

\begin{lemma} \label{lem5} \emph{(Adapted from \cite[Eqn. 3.12]{bay2022optimal})}
    For any $\epsilon \in (0,1)$, when $k = \Theta(n^{\theta})$ for some $\theta \in \big[\frac{1}{2}, 1\big)$ and $T \leq \big(\frac{1}{\log2}k\log_2 \frac{n}{k}\big)(1-\epsilon)$, the number of masked defectives $M_1$ satisfies
    \begin{equation} \label{eq:bay1}
        \bP\left(M_1 \geq \frac{k}{4n}k^{1+\xi}\right) = 1 -o(1),
    \end{equation}
    where $\xi > 0$ is arbitrarily small.
\end{lemma}
In Appendix \ref{sec:adaptations}, we describe how we obtain this result from the results in \cite{bay2022optimal}. 
We can now use the same techniques as in the first case to find a lower bound on the number of masked non-defective items by viewing the combinatorial prior as the previously discussed the two step procedure (with the first step being i.i.d.~with parameter $q' = \frac{k + \sqrt{k}\log n}{n}$).  Recalling that $M'_1 \ge M_1$, Lemma \ref{lem5} gives
\begin{align}
    \label{eq:51} &1 - o(1) = \bP\left(M_{1}' \geq \frac{k}{4n}k^{1 + \xi}\right) \\
    \label{eq:52} &\hspace{-5.9pt}= \sum_{m' = \frac{k}{4n}k^{1 + \xi}}^{n} \bP(M' = m') \bP\left(M'_{1} \geq \frac{k}{4n}k^{1 + \xi} \mid M' = m'\right) \\
    &\hspace{-5.9pt}= \underbrace{\sum_{m' = \frac{k}{4n}k^{1 + \xi}}^{k^{1 + \frac{\xi}{2}}} \bP(M' = m') \bP\left(M_{1}' \geq \frac{k}{4n}k^{1 + \xi} \mid M' = m'\right)}_{\coloneqq \cA_{3}} \nonumber \\
    \label{eq:53} &\hspace{-5.9pt}+ \underbrace{\sum_{m'=k^{1 + \frac{\xi}{2}}+1}^{n} \bP(M' = m') \bP\left(M_{1}' \geq \frac{k}{4n}k^{1 + \xi} \mid M' = m'\right)}_{\coloneqq \cA_{4}}
\end{align}
From here, the steps are essentially the same as Case 1:
\begin{itemize}
    \item Using similar reasoning to $\cA_1$ (see \eqref{eq:40}), we can show that $\cA_{3} = o(1)$ and hence $\cA_{4} = 1-o(1)$; this is shown by noting that when $m' \leq k^{1 + \frac{\xi}{2}}$ it holds that $m'q' \leq \Theta(\frac{k}{n}k^{1 + \frac{\xi}{2}}) = o(\frac{k}{n}k^{1+\xi})$.
    \item We can then deduce (similar to \eqref{eq:44}) that  $M' \geq k^{1 +\frac{\xi}{2}}$ with probability $1-o(1)$.
    \item Since $k^{1 +\frac{\xi}{2}} = \omega(k)$, we observe that having at least $k^{1 +\frac{\xi}{2}}$ masked items implies having at least $\frac{1}{2}k^{1+\frac{\xi}{2}}$ masked non-defectives when $k$ is large enough.
    \item Finally, we drop the assumption that each item appears in at most $n^\xi$ tests; this is done by noting that Lemma \ref{lem5} still holds when ``number of masked items'' is replaced by ``number of masked items appearing in at most $n^{\xi}$ tests'' (see \cite[Procedure 3.1, Step 1]{bay2022optimal}) and using the same reasoning as Case 1.
\end{itemize}



\medskip
\noindent \textbf{\underline{Putting it All Together}}
\medskip

\noindent We have established that there are $\omega(k)$ masked non-defectives with high probability, and at least one masked defective (see Lemmas \ref{lem4} and \ref{lem5}).  We proceed to use this to bound $\bP(\suc) \coloneqq \bP(\hat{S} \supseteq S)$, where $\hat{S}$ is the output of an arbitrary decoding algorithm. To do so, we first recall a fact from \cite{coja2020optimal}.

\begin{lemma} \emph{\hspace*{-3px}\cite[Fact 3.1]{coja2020optimal}} \label{lem3}
    For a given test matrix $\sfX$ and outcome vector $\bfy$, let $\cV(\sfX, \bfy) \subseteq \cS$ be the collection of satisfying sets, i.e., the collection of $s \in \cS$ that would give rise to the test outcome $\bfy$ when using the test matrix $\sfX$. Then, under the combinatorial prior, the posterior distribution of $S$ satisfies $(S \mid \sfX, \bfy) \sim \mathrm{Unif}(\cV(\sfX, \bfy))$.
\end{lemma}
Lemma \ref{lem3} states that the posterior distribution of $S$ given the test outcomes and testing matrix is uniform over the sets that could have led to the observed outcomes. This fact, in conjunction with the previously established high probability bounds on the number of masked items, allows us to bound $\bP(\suc)$ by showing that any estimate of size $(1 + \eta)k$ with contains the true defective set with low probability over the posterior distribution of $S$.

Let $\cT(\sfX, \bfy) \subseteq \cV(\sfX, \bfy)$ be the size-$k$ satisfying sets that are a subset of $\hat{S}$, and further let $\widetilde{\cT}(\sfX, \bfy) \subseteq \cT(\sfX, \bfy)$ be the set of subsets $S'$ for which all of the following conditions hold:
\begin{itemize}
    \item[(i)] $S'$ is a size-$k$ satisfying set;
    \item[(ii)] $S'$ is a subset of $\hat{S}$;
    \item[(iii)] When the items in $S'$ are defective and the remaining items are non-defective, there is at least $1$ masked defective and at least $\frac{1}{2}k^{1+\frac{\xi}{2}}$ masked non-defectives.
\end{itemize}
Note that we use $\frac{1}{2}k^{1+\frac{\xi}{2}}$ in condition (iii) because it is the smaller of the two bounds we derived when we choose $\xi$ to be sufficiently small (the other bound is $\frac{n}{2k}k^{\frac{\gamma}{2}}$ with $k = o(\sqrt n)$, and we will only need this stronger bound later when proving the second part of the theorem). Accordingly, the calculations presented in the previous part of the proof imply that the true defective set $S$ satisfies condition (iii) with probability approaching one, or equivalently
\begin{equation}
    \bP(S \in \cT(\sfX,\bfY) \setminus \widetilde{\cT}(\sfX,\bfY)) = o(1).
\end{equation}
Accordingly, we can write the success probability as follows:
\begin{align}
    \label{eq:55} \bP(\suc) &= \sum_{\sfX, \bfy} P_{\sfX \bfY}(\sfX, \bfy) \bP(\suc \mid \sfX, \bfy) \\
    \label{eq:56} &= \sum_{\sfX, \bfy} P_{\sfX \bfY}(\sfX, \bfy)\sum_{s \in \cT(\sfX, \bfy)} \bP(S = s \mid \sfX, \bfy) \\
    \label{eq:57} &= \sum_{\sfX, \bfy} P_{\sfX \bfY}(\sfX, \bfy)\sum_{s \in \widetilde{\cT}(\sfX, \bfy)} \bP(S = s \mid \sfX, \bfy)  + o(1) \\
    \label{eq:58} &= \sum_{\sfX, \bfy} P_{\sfX \bfY}(\sfX, \bfy) \frac{\abs{\widetilde{\cT}(\sfX, \bfy)}}{\abs{\cV(\sfX, \bfy)}} + o(1),
\end{align}
where \eqref{eq:58} follows from Lemma \ref{lem3}. 

Now suppose that for a given pair $(\sfX, \bfy)$, the decoder's estimate $\hat{S} = \hat{S}(\sfX, \bfy)$ satisfies $\abs{\widetilde{\cT}(\sfX, \bfy)} = N$.  If $N=0$ then the fraction in \eqref{eq:58} is trivially 0, so we henceforth assume that $N > 0$. 
%
%
We refer to the elements of $\widetilde{\cT}(\sfX, \bfy)$ as being \emph{plausible} for $\hat{S}$, and the elements of $\cV(\sfX, \bfy) \setminus \widetilde{\cT}(\sfX, \bfy)$ as being \emph{implausible}.  The number of plausible sets is $N$, and we will show that $\frac{\abs{\widetilde{\cT}(\sfX, \bfy)}}{\abs{\cV(\sfX, \bfy)}} = o(1)$ by arguing that the number of implausible sets is $\omega(N)$.

Denote the plausible sets by $S_1,\dotsc,S_N$ (the dependence on $(\sfX, \bfy)$ is left implicit), and note that by definition $|S_i|=k$.  We first argue that for every plausible set, there are many implausible sets obtained by swapping out a single element.  Specifically, swapping out one masked defective item (in $S_i$) with one masked non-defective item in $[n] \setminus \hat{S}$ produces an implausible set, because being plausible requires being a subset of $\hat{S}$.  Note that $S_i \subseteq \hat{S}$, so combining $|\hat{S}| = (1+\eta)k$ and $|S_i| = k$ gives that there are at most $\eta k$ masked non-defectives in $\hat{S}$, and hence at least  $\frac{1}{2}k^{1+\frac{\xi}{2}} - \eta k$ masked non-defectives in $[n] \setminus \hat{S}$ (see the definition of $\widetilde{\cT}$).  Thus, there are at least $\frac{1}{2}k^{1+\frac{\xi}{2}} - \eta k$ ways of performing the swap described above. For convenience, we apply $\eta = k^{o(1)}$ to simplify $\frac{1}{2}k^{1+\frac{\xi}{2}} - \eta k \geq \frac{1}{4}k^{1+\frac{\xi}{2}}$ for sufficiently large $k$, and work with this simpler lower bound.

When we apply this argument to each plausible set $S_1,\dotsc,S_N$, some care is needed because their associated $N$ implausible sets  may overlap.  To upper bound the \emph{total} number of implausible sets, we identify a sufficiently large subset $\{S_i\}_{i=1}^{N'}$ for some $N' < N$ whose associated implausible sets (as constructed above) have no overlap.  For brevity, we refer to $S_i$'s associated implausible sets as the \emph{neighbors} of $S_i$.  Importantly, while the sets $S_i$ will typically have \emph{multiple} masked defectives, we designate a \emph{single} one (say, the one with the smallest index) to be the one that gets swapped out in the previous paragraph.

Then, we claim that in order for two plausible $S_i$ and $S_j$ to share a common neighbor, the following conditions must all hold:
\begin{itemize}
    \item $|S_i \setminus S_j| = |S_j \setminus S_i| = 1$;
    \item The unique element of $S_i \setminus S_j$ is the designated masked defective for $S_i$;
    \item The unique element of $S_j \setminus S_i$ is the designated masked defective for $S_j$.
\end{itemize}
The first condition is required because if $|S_i \setminus S_j| \ge 2$, then it will be impossible to swap out both of these 2 (or more) items (the neighbors only consist of performing a single swap), and similarly if $|S_j \setminus S_i| \ge 2$.  The two remaining conditions are required because if the element of $S_i \setminus S_j$ is not the designated defective, it is not allowed to be swapped out from $S_i$ (similarly for $S_j$).

Then, we consider the following greedy construction analogous to the Gilbert-Varshamov construction \cite{gilbert1952comparison}:
\begin{enumerate}
    \item Start with the empty set.
    \item Choose an arbitrary $i \in \{1,2, \dots, N\}$, add $S_i$ to the set being constructed, and then remove $S_i$ from further consideration along with all other $ \{S_j\}_{j \ne i} $ that satisfy the three dot points above.
    \item Repeat step 2 with the remaining plausible sets, and continue until there are no plausible sets remaining.
\end{enumerate}
We claim that for each $S_i$ chosen, the number of $S_j$ ($j \ne i$) removed is at most $\eta k$, where we recall that $\eta = k^{o(1)}$.  This is because $S_i$ and $S_j$ necessarily share $k-1$ items in common, and there is only one choice for the item in $S_i \setminus S_j$ (it must be $S_i$'s designated defective) and at most $\eta k$ choices for the item in $S_j \setminus S_i$ (it must be an element of $\hat{S} \setminus S_i$, where $|\hat{S}|=(1+\eta)k$ and $|S_i|=k$).

Thus, the above procedure extracts at least $\frac{N}{\eta k + 1}$ plausible sets, all of which have pairwise disjoint collections of implausible sets. As we already argued, each such collection contains at least $\frac{1}{4}k^{1+\frac{\xi}{2}}$ implausible sets, implying that $\abs{\cV(\sfX, \bfy)} \geq \frac{N}{\eta k + 1}\frac{1}{4} k^{1+\frac{\xi}{2}} + N$.  Since $\eta = k^{o(1)}$, this implies that $\frac{\abs{\widetilde{\cT}(\sfX, \bfy)}}{\abs{\cV(\sfX, \bfy)}} \leq \frac{N}{\frac{N}{\eta k + 1}\frac{1}{4} k^{1+\frac{\xi}{2}} + N} = \frac{1}{\Omega\big(k^{\frac{\xi}{2} - o(1)}\big)} = o(1)$.  Substituting this into \eqref{eq:58}, we obtain $\bP(\suc) = o(1)$ as claimed. 

\medskip
\noindent \textbf{\underline{Establishing the Second Part of the Theorem}}
\medskip

\noindent When $\theta \in \big(\frac{\log 2}{1+\log 2}, \frac{1}{2}\big)$, we have established the stronger claim that there are $\omega(\frac{n}{k})$ masked non-defective items with high probability (rather than only $\omega(k)$).  We claim that this implies the same result when $\eta = \Omega(k^{\lambda})$ for $\lambda \in [0, \frac{1}{\theta}-2)$ (rather than only $\eta = k^{o(1)}$).  This is because under such a choice, we have (via $k = \Theta(n^{\theta})$) that the number of possible swaps with items not in $\hat{S}$ is at least $\frac{n}{2k}k^{\frac{\gamma}{2}} - \eta k = \Theta(\frac{n}{k}k^{\frac{\gamma}{2}})$, and that $\frac{(n/k)k^{\frac{\gamma}{2}}}{\eta k} \to \infty$, which implies $\frac{\abs{\widetilde{\cT}(\sfX, \bfy)}}{\abs{\cV(\sfX, \bfy)}} = o(1)$ by the same reasoning as that above. \qed

\section{Conclusion} \label{sec:conclusion}

We have established optimal rates for group testing with one-sided approximate recovery; specifically, for SUBSET we established an algorithm whose rate matches the counting bound, and for SUPERSET we established a converse that matches the better of two existing achievability results.  

An immediate direction for further work would be to attain the optimal rate for SUBSET with \emph{polynomial decoding time}, possibly building on the spatial coupling approach to exact and two-sided approximate recovery \cite{coja2020optimal}. While such an approach gives us the required starting point of a two-sided approximate recovery estimate, the subsequent analysis in Section \ref{sec:ach_subset_pf} may be difficult to adapt to the more complicated test design.  Perhaps more fundamentally, even given the initial estimate, our subsequent search step is still over $n^{\omega(1)}$ candidate sets, so the question remains of how to maintain its performance with polynomial runtime.

As discussed in Remark \ref{rem:lambda}, another possible direction is to establish the best possible values of $\lambda$ that can be attained in results of the kind stated in Theorems \ref{thm1} to \ref{thm3}.

\appendix
\subsection{Concentration Inequalities} \label{app:bino}
Here we provide statements of some standard concentration inequalities (e.g., see \cite[Ch. 2]{concentration}) that we use throughout the paper. Letting $X_{1}, X_{2}, \dots, X_{n}$ be a sequence of i.i.d.~Bernoulli($\mu$) random variables, we have the following:
\begin{itemize}
    \item (Chernoff Bound for Binomial RVs) For any $\delta > 0$, we have 
    \begin{align}
        &\bP\left(\sum_{i=1}^{n}X_{i} \geq (1+\delta)n\mu\right) \nonumber \\
        \label{eq:chernoff_large} &\qquad \qquad \leq \exp\Big(-n\mu\big((1+\delta)\log(1+\delta)-\delta\big)\Big),
    \end{align}
    and for any $\delta \in (0,1]$, we have
    \begin{align}
        &\bP\left(\sum_{i=1}^{n}X_{i} \leq (1-\delta)n\mu\right) \nonumber \\
        \label{eq:chernoff_small} &\qquad \qquad \leq \exp\Big(-n\mu\big((1-\delta)\log(1-\delta)+\delta\big)\Big).
    \end{align}
    \item (Weakened Chernoff Bound for Binomial RVs) For any $\delta \in (0,1]$, we have 
    \begin{equation}
        \label{eq:weakchernoff_large} \bP\left(\sum_{i=1}^{n}X_{i} \geq (1+\delta)n\mu\right) \leq \exp\Big(-\frac{\delta^{2}}{3}n\mu\Big),
    \end{equation}
    and for any $\delta \in (0,1]$, we have
    \begin{equation}
        \label{eq:weakchernoff_small} \bP\left(\sum_{i=1}^{n}X_{i} \leq (1-\delta)n\mu\right) \leq \exp\Big(-\frac{\delta^{2}}{2}n \mu\Big).
    \end{equation}
\end{itemize}

\subsection{Two-Sided Approximate Recovery Rate  with Smaller Distortion} \label{sec:two_sided}

Here we outline the proof that algorithm in the algorithm in \cite{scarlett2016phase} still succeeds for two-sided approximate recovery when $\beta = \Omega(k^{-\lambda(\frac{1}{\theta}-1)})$ and $\lambda \in \big[0,\frac{1}{2}\big)$, as opposed to only constant $\beta \in (0,1)$.  We note that a variant of this idea (with $\beta  = \Theta(k^{\gamma - 1})$ for $\gamma \in (0,1)$) has been analyzed in \cite{scarlett2018noisy}, but the details differ slightly since they consider the noisy case.

Recall the condition in \eqref{eq:12} defining the decoder, and let $T I(\tau)$ be the mean of $\imath^{T}\left(\sfX_{\sdif} ; \bfy \mid  \sfX_{\seq}\right)$.  (Here $I(\tau)$ is a conditional mutual information, but the precise definitions of $\imath^{T}$ and $I(\tau)$ are not needed for our purposes.)  
The following condition on the number of tests is stated in \cite[Eqn.~3.37]{scarlett2016phase} for parameters $\delta_1$ and $\{\delta_{2,\tau}\}_{\tau = \lfloor\beta k + 1\rfloor}^{k}$:
\begin{equation}
    \label{eq:T_bound} T \geq \frac{\log{n-k \choose \tau}+\log \big(\frac{k}{\delta_{1}}{k \choose \tau}\big)}{I(\tau)(1-\delta_{2, \tau})}.
\end{equation}
Letting $\cE$ denote the failure event for two-sided approximate recovery, \cite[Eqn.~3.38]{scarlett2016phase} states that if \eqref{eq:T_bound} holds for all $\tau = \lfloor\beta k+1\rfloor,\dots, k$, then the error probability is upper bounded as follows:
\begin{align} \label{eq:two_sided_error}
    \bP(\cE) \leq \sum_{\tau =\lfloor \beta k+1\rfloor}^{k}{k \choose \tau} \psi_{\tau}(T, \delta_{2, \tau}) + \delta_1,
\end{align}
where $\psi_{\tau}: \bZ \times \bR \to \bR$ comes from a concentration inequality given in \cite[Eqn. 3.45]{scarlett2016phase}, and is defined by
\begin{align}
    &\psi_{\tau}(T, \delta_{2, \tau}) = \nonumber \\
    &\begin{cases}
        \exp\Big(-T \frac{\tau}{2k}\log 2 \Big((1 -\delta_{2}^{(1)}) \log(1 - \delta_{2}^{(1)}) \\ \qquad  \qquad \qquad \qquad \quad+ \delta_{2}^{(1)}\Big)(1 - \xi)\Big) & \hspace{-8.1pt}\tau \leq \lfloor\frac{k}{\log k}\rfloor \\
         2\exp\left(-\frac{(\delta_{2}^{(2)}I(\tau))^{2}T}{4(8+\delta_{2}^{(2)}I(\tau))}\right) & \hspace{-8.1pt}\tau > \lfloor \frac{k}{\log k}\rfloor       
    \end{cases} \label{eq:psi_def}
\end{align}
where $\xi > 0$ is arbitrarily small.

We apply \eqref{eq:two_sided_error} with $\delta_1 = \frac{1}{k}$, so that the second term is insignificant.  As for each $\delta_{2, \tau}$, this is set to $\delta_{2, \tau} = \delta_{2}^{(1)}$ for $\tau \leq \lfloor \frac{k}{\log k}\rfloor$, and to $\delta_{2, \tau} = \delta_{2}^{(2)}$ for $\tau > \lfloor\frac{k}{\log k}\rfloor$, where $\delta_{2}^{(1)}$ and $\delta_{2}^{(2)}$ will be specified shortly. We can then split the summation in \eqref{eq:two_sided_error} as follows:
\begin{align}
    \bP(\cE) &\leq \sum_{\tau = \lfloor \beta k + 1\rfloor}^{\lfloor \frac{k}{\log k}\rfloor}{k \choose \tau}\psi_{\tau}(T,\delta_{2}^{(1)}) \nonumber \\
    &+\sum_{\tau = \lfloor \frac{k}{\log k}\rfloor + 1}^{k}{k \choose \tau}\psi_{\tau}(T, \delta_{2}^{(2)}) + \frac{1}{k}. \label{eq:two_sums}
\end{align}
We can directly use part of the analysis in \cite[Appendix B]{scarlett2016phase}, showing that if $T \geq \big(k \log_{2}\frac{n}{k}\big)(1+\epsilon)$, then (i) \eqref{eq:T_bound} is satisfied for all $\tau= \lfloor \beta k + 1\rfloor,\dots, k$, and (ii) the second sum in \eqref{eq:two_sums} tends to $0$ for arbitrarily small $\delta_{2}^{(2)} > 0$.  Hence, to show that $\bP(\cE) \to 0$, it suffices to show that the first sum in \eqref{eq:two_sums} also decays to zero. Note that if $\beta k + 1 > \frac{k}{\log k}$ then the sum is empty and is trivially zero, so henceforth we assume this is not the case. We proceed to bound this first sum as follows via \eqref{eq:psi_def}:
\begin{align}
    \label{eq:sum_bound_1}\cA_{} &\coloneqq \sum_{\tau = \lfloor \beta k + 1\rfloor }^{\lfloor \frac{k}{\log k}\rfloor}{k \choose \tau}\psi_{\tau}(T,\delta_{2}^{(1)}) \\
    &= \sum_{\tau = \lfloor \beta k + 1 \rfloor}^{\lfloor \frac{k}{\log k}\rfloor}{k \choose \tau}\exp\Big(-T \frac{\tau}{2k}\log 2 \nonumber \\
    \label{eq:sum_bound_2} &\qquad \qquad \times \big((1 -\delta_{2}^{(1)}) \log(1 - \delta_{2}^{(1)}) + \delta_{2}^{(1)}\big)(1 - \xi)\Big) \\
    &\leq k\underset{\lfloor\beta k + 1\rfloor \leq \tau\leq \lfloor \frac{k}{\log k}\rfloor}{\max} {k \choose \tau}\exp\Big(-T\frac{\tau}{2k}\log 2\nonumber \\
    \label{eq:sum_bound_3} & \qquad \qquad \times \big((1 - \delta_{2}^{(1)}) \log(1 - \delta_{2}^{(1)}) + \delta_{2}^{(1)}\big)(1 - \xi)\Big) \\
    &=\underset{\lfloor \beta k+1\rfloor \leq \tau \leq \lfloor \frac{k}{\log k}\rfloor}{\max} \exp\Big(\log k + \log {k \choose \tau} - T \frac{\tau}{2k} \nonumber \\
    \label{eq:sum_bound_4} &\qquad \times\log 2 \big((1 - \delta_{2}^{(1)})\log(1 - \delta_{2}^{(1)})+ \delta_{2}^{(1)}\big)(1 - \xi)\Big) \\
    & \leq \underset{\lfloor \beta k+1 \rfloor \leq \tau \leq \lfloor \frac{k}{\log k}\rfloor}{\max} \exp\Big(\log k + \log{k \choose \tau} \nonumber \\
    \label{eq:sum_bound_5} &\qquad \qquad \qquad  \qquad \qquad \qquad -T\frac{\tau}{2k}\cdot\frac{\log 2}{1+\epsilon'}\Big),
\end{align}
where \eqref{eq:sum_bound_5} follows by taking $\delta_{2}^{(1)}$ to be arbitrarily close to $1$ and introducing $\epsilon'  > 0$ that can be made arbitrarily small (as a function of $\xi$ and $\delta_{2}^{(1)}$). In order for $\cA \to 0$, we require that the exponent in \eqref{eq:sum_bound_5} approaches $-\infty$.  For this to happen, it suffices that the following holds for all $\tau \in \big\{\lfloor \beta k + 1\rfloor,\dotsc,\lfloor \frac{k}{\log k}\rfloor\big\}$:
\begin{align}
   \label{eq:test_bound} T \geq \left(\log k + \log {k \choose \tau}\right)\frac{2k}{\tau\log 2 }(1 + 2\epsilon').
\end{align}
Since $\tau = o(k)$, we can write $\log {k \choose \tau} = \big(\tau\log\frac{k}{\tau}\big)\big(1+o(1)\big)$, and substituting this into \eqref{eq:test_bound} gives
\begin{align}
    \label{eq:test_bound_2.1} T &\geq \Big(\log k + \Big(\tau \log \frac{k}{\tau}\Big)\big(1+o(1)\big)\Big)\frac{2k}{\tau\log 2}(1+2\epsilon') \\
   \label{eq:test_bound_2.2} & = \frac{2}{\log 2}\Big(k \log \frac{k}{\tau}\Big)(1+2\epsilon')\big(1+o(1)\big) \\
   \label{eq:test_bound_2.3} &= 2 \Big( k \log_{2}\frac{k}{\tau}\Big)(1+2\epsilon')\big(1+o(1)\big).
\end{align}
Observe that the strictest requirement on $T$ is when $\tau = \lfloor\beta k + 1\rfloor$. Substituting this into \eqref{eq:test_bound_2.3} and recalling that $\beta = \Omega((\frac{k}{n})^\lambda)$, the bound simplifies to
\begin{align}
    \label{eq:test_bound_2.4} T &\geq 2k \Big(\log_{2}\Big(\Big(\frac{n}{k}\Big)^\lambda\Big)\Big)(1+2\epsilon')\big(1+o(1)\big) \\ 
    \label{eq:test_bound_2.5} &= 2\lambda k\Big(\log_{2}\frac{n}{k}\Big)(1+2\epsilon')\big(1+o(1)\big).
\end{align}
By the assumption $\lambda \in [0,\frac{1}{2})$, this condition on $T$ is satisfied when $T = (k\log_{2}\frac{n}{k})(1+\epsilon)$ for arbitrarily small $\epsilon > 0$.  Having established that $\bP(\cE) \to 0$ under this condition, the proof is complete. \qed

\subsection{Adaptations of Existing Intermediate Results (Lemmas \ref{lem4}-\ref{lem5})} \label{sec:adaptations}

In this section, we provide details on the modifications we make to the results in \cite{coja2020optimal, bay2022optimal} to obtain Lemmas \ref{lem4}-\ref{lem5}.  
Throughout the section, we will use symbols with a tilde (e.g., $\tilde{n}$)  to denote quantities in a ``denser'' scaling regime, which will then be related to ``sparser'' regimes with the usual notation (e.g., $n$).  When the two settings both have the same value of a certain parameter, the tilde will be omitted (e.g., $k$).

\subsubsection{Establishing Lemma \ref{lem4}} \label{sec:lem4_proof}

The relevant analysis in \cite{coja2020optimal} first bounds the number of masked defectives under the i.i.d.~prior in a (relatively) dense setting, then translates it to the combinatorial prior in a dense setting, and finally to the combinatorial prior in a sparse setting.  The details are given as follows.

Consider a setting with population size $\tilde{n}$ in which each item is defective independently with probability $\tilde{q}' = \frac{k - \sqrt{k}\log \tilde{n}}{\tilde{n}}$, where $k = \tilde{n}^{\theta'}$ for some $\theta'$ sufficiently close to one such that $2(1-\theta') < \xi <\epsilon\theta'$ for arbitrarily small $\xi > 0$.  Under this setting, it is shown in \cite[Proof of Corollary 3.13]{coja2020optimal} that the number of masked defectives $\widetilde{M}_{1}'$ stochastically dominates $\bin(\tilde{n}^{1-\xi}, \frac{1}{3}\tilde{n}^{\epsilon \theta' -1})$. Accordingly, we have
\begin{align}
    &\bP\Big(\widetilde{M}_{1}' \leq \frac{1}{4}\tilde{n}^{\epsilon \theta' - \xi}\Big) \nonumber \\
    \label{eq:stoch_dom_1} &\leq \bP\Big(\bin\Big(\tilde{n}^{1-\xi}, \frac{1}{3} \tilde{n}^{\epsilon \theta' -1}\Big) \leq \frac{1}{4}\tilde{n}^{\epsilon\theta' - \xi}\Big) = o(1),
\end{align}
where the last step uses standard binomial concentration (Appendix \ref{app:bino}). The analysis in \cite{coja2020optimal} then relates this result to the number of masked defectives under the combinatorial prior (in a population of size $\tilde{n}$) via the following two-step generation of $S$:
\begin{enumerate}
    \item Generate $S' \subseteq [\tilde{n}]$ by including each item in it independently with probability $\tilde{q}' = \frac{k - \sqrt{k}\log \tilde{n}}{\tilde{n}}$
    \item Add $\max\{k - \abs{S'}, 0\}$ items from $[\tilde{n}] \setminus S'$ to $S'$ uniformly at random to form $S$.
\end{enumerate}
This two-step procedure resembles the one used after Lemma \ref{lem4}, except we are now adding defective items from the first to the second step rather than removing them.  The first step uses the i.i.d.~prior, meaning the high probability bound in \eqref{eq:stoch_dom_1} applies to the number of masked defectives in $S'$. Since the second step involves making more items defective, the number of masked items following this step can only increase, meaning that $\widetilde{M}_1'$ serves as a lower bound on the number of masked defectives under this distribution. Moreover, under the high-probability event $\{k - 2 \sqrt{k}\log n\leq \abs{S'} \leq k\}$, the resultant distribution on $S$ follows the combinatorial prior, meaning any bounds we derive under this distribution of $S$ also hold for the combinatorial prior up to $o(1)$ events. 

Letting $\widetilde{M}_1$ denote the number of masked defectives under the above-mentioned combinatorial prior, the above discussion and \eqref{eq:stoch_dom_1} imply that $\widetilde{M}_1 \geq \widetilde{M}_1' \geq \frac{1}{4}n^{\epsilon \theta' - \xi}$ with probability $1-o(1)$. Since rearranging $k = \tilde{n}^{\theta'}$ gives $\tilde{n} = k^{\frac{1}{\theta'}}$, we have
\begin{equation}
    \label{eq:modified_stoch_dom} 
    \widetilde{M}_{1} \geq \frac{1}{4}\tilde{n}^{\epsilon \theta' - \xi} \iff \widetilde{M}_{1} \geq \frac{1}{4}k^{\epsilon - \frac{\xi}{\theta'}}.
\end{equation}
Thus, with this parametrization, there are at least $\widetilde{M}_{1} \ge \frac{1}{4}k^{\epsilon - \frac{\xi}{\theta'}}$ masked defective items, with probability $1-o(1)$. 

We now recap how converse results can be translated from denser settings to sparser settings, i.e., from higher $\theta'$ (near 1) to smaller $\theta$.  This is established in \cite{coja2020optimal} via the contrapositive statement that achievability for small $\theta$ implies the same for large $\theta'$ (with the same number of tests, say $T = c k \log k$ for some $c > 0$ when the two problems have the same $k$ but differing population sizes). To show this, the idea is to start with a test design $\sfX$ for the sparser problem with parameters $n$ and $\theta$, and construct a test design $\sfX '$ for the denser problem with parameters $\tilde{n} < n$ and $\theta ' > \theta$ as follows:
\begin{enumerate}
    \item Choose a size-$\tilde{n}$ subset $V \subset [n]$ uniformly at random from the ${n \choose \tilde{n}}$ possibilities, and declare the remaining items $[n] \setminus V$ to be (dummy) non-defectives.
    \item Map the $\tilde{n}$ denser problem's items to $V$ in an arbitrary fixed manner (e.g., in numerical order).
    \item Perform the tests (and retrieve their results $\bfY'$) according to $\sfX'$, defined to be the sub-matrix of the sparser problem's test matrix $\sfX$ with the columns indexed by $V$.
\end{enumerate}
We observe that since the items in $[n] \setminus V$ are non-defective, the test results $\bfY$ and $\bfY'$ generated by $\sfX$ and $\sfX'$ respectively are the same (since the items removed to form $\sfX '$ are designated as non-defective).  Moreover, due to the symmetry of the construction, it is easy to see that if the denser problem has $k$ uniformly random defectives, then so does the sparser problem (accounting for the randomness in both the defective set and in $V$).  Combining the findings from the preceding two sentences, we have that an achievability result for the sparser problem immediately implies the same for the denser problem.
%
We refer the interested reader to \cite{coja2020optimal} for the full details.

We can use a similar idea for our problem setup. Namely, we can transfer our results from the high $\theta$ regime to the lower $\theta$ regime for $\theta \in \big(\frac{\log 2}{1+\log 2}, \frac{1}{2}\big)$ by adding in these dummy non-defectives, and noting that doing so can only increase the number of masked non-defectives, and leaves the number of masked defectives unchanged (since non-defective items do not affect masking events of defectives). This means that $\widetilde{M}_{1} = M_{1}$, and \eqref{eq:30} follows by combining this with \eqref{eq:stoch_dom_1}--\eqref{eq:modified_stoch_dom}.


\subsubsection{Establishing Lemma \ref{lem:degrees}}

Fix a test $t \in \{1,2,\dots, T\}$ containing $\gamma \geq \Gamma = \frac{n}{k} \log n$ items, and define the event $\cE_t \coloneqq \{\lvert \cI_t \cap S\rvert = 0\}$ (i.e., the test outcome is negative), where we recall that $\cI_t \subseteq [n]$ is the set of items included in test $t$. Without loss of generality, we assume that $\gamma = \Gamma$, since adding more items to the test can only decrease the probability of a negative test. Thus under the combinatorial prior, the probability of $\cE_t$ occurring (with respect to the randomness in $S$) is bounded by
\begin{align}
    \label{eq:neg_test_prob_1} \bP(\cE_t) &= \frac{\lvert \{s \in \cS : \lvert \cI_t \cap s \rvert = 0\}\rvert}{{n \choose k}} \\
    \label{eq:neg_test_prob_2} &= \frac{{n - \Gamma \choose k}}{{n \choose k}} \\
    \label{eq:neg_test_prob_3} &= \frac{(n-\Gamma)(n-\Gamma -1)\cdots (n-\Gamma - k+1)}{n(n-1)\cdots(n-k+1)} \\
    \label{eq:neg_test_prob_4} &\leq \left(\frac{n- \Gamma}{n}\right)^k \\
    \label{eq:neg_test_prob_5} &\leq \exp\left(-\frac{k}{n}\Gamma\right) \\
    \label{eq:neg_test_prob_6} &= \frac{1}{n},
\end{align}
where:
\begin{itemize}
    \item \eqref{eq:neg_test_prob_1} uses $S \sim {\rm Unif}(\cS)$;
    \item \eqref{eq:neg_test_prob_2} follows since any set $s$ such that $\lvert \cI_t \cap s \rvert = 0$ must be a subset of $[n] \setminus \cI_t$. Since $\lvert \cI_t\rvert = \Gamma$, there are ${n - \Gamma \choose k}$ such sets;
    \item \eqref{eq:neg_test_prob_3} expands out the binomial coefficients and cancels terms;
    \item \eqref{eq:neg_test_prob_4} follows by upper bounding $\frac{n - \Gamma - i}{n-i}$ by $\frac{n-\Gamma}{n}$ for each $i = 1, \dots, k-1$;
    \item \eqref{eq:neg_test_prob_5} uses $1-x \leq e^{-x}$ for $x \geq 0$, and applies this with $x = \frac{\Gamma}{n}$;
    \item \eqref{eq:neg_test_prob_6} follows by substituting the choice of $\Gamma$ and simplifying.
\end{itemize}
Taking a union bound over the at most $T = O(k \log n) = o(n)$ tests containing at least $\Gamma$ items yields the result.

Next, we show that if a test design $\sfX$ places at most $\Gamma$ items in each test, then the resultant number of items in at most $\log^3 n$ tests is $n - o(n)$. For each $i \in [n]$, define $\cT_i \subseteq \{1,2,\dots, T\}$ to be the set of tests in which $i$ is included. A double counting argument then yields
\begin{equation}
    \label{eq:double_count} \sum_{i=1}^n \lvert \cT_i\rvert = \sum_{t=1}^T \lvert \cI_t\rvert \leq  T \cdot \Gamma = O(n \log^2 n),
\end{equation}
where the last inequality follows since $T = O(k \log n)$. Thus, the number of items that appear in more than $\log ^3n$ tests is at most $O(\frac{n}{\log n}) = o(n)$, which gives the desired result.

\subsubsection{Establishing Lemma \ref{lem5}}

For Lemma \ref{lem5}, we will ultimately use the same ``dummy non-defectives'' used in proving  Lemma \ref{lem4}, but before doing so, we need to transfer a result from \cite{bay2022optimal} from \emph{number of masked items} to \emph{number of masked defectives}.  

We begin by using \cite[Eq.~3.12]{bay2022optimal} and its preceding paragraph, which state that when using an i.i.d.~prior with parameter $\tilde{q} = \frac{k}{\tilde{n}}$ and $\tilde{n}$ items, where $k = \tilde{n}^{\theta'}$ for $\theta' = 1 - o(1)$, the number of masked items stochastically dominates $\bin(\tilde{n}^{1-3\xi}, e^{-\frac{(1+\xi)c_{\tilde{q}}}{\tilde{n}\tilde{q}}T})$ for some arbitrarily small $\xi > 0$ and constant $c_{\tilde{q}} > 0$ that depends on $\tilde{q}$. Moreover, when $\tilde{q} = o(1)$ (which is the case we consider), this constant is given explicitly by $c_{\tilde{q}} = (\log^{2} 2)\big(1+o(1)\big)$.  To obtain a similar result for masked \emph{defectives}, it is useful to note how the binomial distribution arises in \cite{bay2022optimal}:
\begin{itemize}
    \item Using the same ideas as \cite{coja2020optimal}, they identify a subset of $\tilde{n}^{1-3\xi}$ items that are masked with probability at least $e^{-\frac{(1+\xi)c_{\tilde{q}}}{\tilde{n}\tilde{q}}T}$, and whose \emph{masking events are independent};
    \item The independence of masking events is ensured by a \emph{distance separation property}, which can be stated as follows: For any two items $i$ and $i'$ among the subset of size $\tilde{n}^{1-3\xi}$, the following two sets are disjoint: (1) items that share any test with $i$ (including $i$ itself); (2) items that share any test with $i'$ (including $i'$ itself).
\end{itemize}
To adapt the desired result from masked items to masked defectives, we simply observe that this separation property ensures that the property of being \emph{both masked and defective} is still independent across these items.\footnote{In fact, an equivalent property was implicitly used in \cite[Eq.~(3.14)]{bay2022optimal} later in their analysis, and this was also used in \cite{coja2020optimal} to establish the binomial distribution mentioned in the proof of Lemma \ref{lem4} above.}  Combining this with the fact that any masked item is defective with probability matching the prior (i.e., $\tilde{q}$) \cite{coja2020optimal,bay2022optimal}, it follows that the number of masked defectives stochastically dominates  $\bin(\tilde{n}^{1-3\xi}, \tilde{q}\cdot e^{-\frac{(1+\xi)c_{\tilde{q}}}{\tilde{n}\tilde{q}}T})$.  Observe that this matches the above-mentioned bound for masked items, but with the probability multiplied by $\tilde{q}$.


Using binomial concentration (Appendix \ref{app:bino}) and $c_{\tilde{q}} = (\log^{2} 2)\big(1+o(1)\big)$, it follows that the number of masked defectives (under an i.i.d.~prior with probability $\tilde{q} = \frac{k}{\tilde{n}}$) is at least
\begin{equation}
    \label{eq:bay_mean} \frac{1}{2} \tilde{q} \cdot \tilde{n}^{1-3\xi} \exp\Big(-\frac{(1+\xi)\log^{2}2}{k}T\big(1+o(1)\big)\Big)
\end{equation}
with probability $1-o(1)$, provided that \eqref{eq:bay_mean} itself has $\omega(1)$ scaling (which will be true under our choice of $T$ below).  Note that we have replaced $\tilde{n}\tilde{q}$ by $k$ in the exponent since the two are identical.
Then, as noted in \cite[Section 3.2]{bay2022optimal}, the exponent in \eqref{eq:bay_mean} remains the same when the prior probability $\tilde{q} = \frac{k}{\tilde{n}}$ is replaced by $\tilde{q}' = \frac{k - \sqrt{k}\log \tilde{n}}{\tilde{n}}$, since any changes are only in the $1+o(1)$ term. Letting $\widetilde{M}'_1$ denote the number of masked defectives when each item is defective independently with probability $\tilde{q}'$, we deduce that the following holds with probability $1-o(1)$:
\begin{align}
    \label{eq:mtilde_bound_1} \widetilde{M}_1' &\geq \frac{1}{2} \tilde{q}' \cdot \tilde{n}^{1-3\xi} \exp\Big(-\frac{(1+\xi)\log^{2}2}{k}T\big(1+o(1)\big)\Big) \\
    \label{eq:mtilde_bound_2} &\geq \frac{1}{4}\cdot \frac{k}{\tilde{n}}\tilde{n}^{1-3\xi} \exp\Big(-\frac{(1+\xi)\log^{2}2}{k}T\big(1+o(1)\big)\Big),
\end{align}
where \eqref{eq:mtilde_bound_2} uses the fact that $\frac{k - \sqrt{k}\log \tilde{n}}{\tilde{n}} \geq \frac{k}{2\tilde{n}}$ for sufficiently large $\tilde{n}$. 

This bound can then be transferred to a bound on the number of masked defective items $\widetilde{M}_1$ under the combinatorial prior via the same argument involving the two-step generation of $S$ detailed in the proof of Lemma \ref{lem4} above. We omit the details to avoid repetition, and simply write the resulting lower bound:
\begin{equation} \label{eq:m1_bound_2}
    \widetilde{M}_{1} \geq \frac{1}{4}\cdot\frac{k}{\tilde{n}}\tilde{n}^{1-3\xi}\exp\Big(-\frac{(1+\xi)\log^{2}2}{k}T\big(1+o(1)\big)\Big)
\end{equation}
with probability $1-o(1)$. Writing $\tilde{n} = k^\frac{1}{\theta'}$ and applying some algebraic simplifications, \eqref{eq:m1_bound_2} yields that with probability $1-o(1)$, the number of masked defectives $\widetilde{M}_{1}$ is at least
\begin{equation} \label{eq:m1_bound_3}
     \frac{1}{4} k^{1-\frac{3\xi}{\theta'}}\exp\Big(-\frac{(1+\xi)\log^{2}2}{k}T\big(1+o(1)\big)\Big).
\end{equation}
With this finding in place, we use the same ``dummy non-defectives'' idea as that of Lemma \ref{lem4} to make it such that there are $k = \Theta(n^{\theta})$ defective items in a population of size $n$ (but now with $\theta \in \big[\frac{1}{2}, 1\big)$), which leaves the number of masked defectives unchanged as already argued, i.e., $\widetilde{M}_1 = M_1$. Letting $T = \big(\frac{1}{\log 2}k\log_2\frac{n}{k}\big)(1-\epsilon)$ for some arbitrarily small suitably chosen $\epsilon > 0$, \eqref{eq:m1_bound_3} becomes $\frac{k}{4n}k^{1+\xi}$, thus yielding \eqref{eq:bay1} in Lemma \ref{lem5}.

\subsection{Proof of Theorem \ref{thm1} (Approximate Recovery Bounds for COMP and DD)} \label{sec:comp_dd}

The term $\zeta(\theta)$ is the well-documented exact recovery rate \cite{coja2020optimal}, so it remains to show that the rate $\log 2$ is achievable for both SUBSET and SUPERSET.  This is done using the \emph{Definite Defectives} (DD) algorithm \cite{aldridge2014group} for SUBSET, and the \emph{Combinatorial Orthogonal Matching Pursuit} (COMP) algorithm \cite{chan2011non} for SUPERSET. In both cases, we adopt the near constant column weight design \cite{johnson2018performance} where each item is placed in $L = \frac{T \log 2}{k}$ tests independently with replacement.

An outline of achievable approximate recovery rates using DD/COMP is given in \cite[Section 5.1]{aldridge2019group} for the case that $\eta^-,\eta^+$ are fixed constants in $(0,1)$.  However, we are unaware of any works giving full proofs, and moreover, the result that we seek in Theorem \ref{thm1} also permits smaller values of $\eta^-,\eta^+$, which is non-trivial to achieve.  Thus, we provide a detailed proof for completeness, despite amounting to relatively straightforward extensions of existing work. As with Section \ref{sec:proofs}, we subsequently drop the superscripts from $\eta^-, \eta^+$ since they will be clear from the context.

\subsubsection{Achievability for SUPERSET Using COMP}
The COMP algorithm declares all items that appear in at least one negative test as non-defective, and all other items as defective.  Since defectives can never appear in negative tests, it follows that every defective item is correctly classified, and the resulting estimate $\hat{S}_{\mathrm{COMP}}$ satisfies $\hat{S}_{\mathrm{COMP}} \supseteq S$.  To show that a rate of $\log 2$ is achievable, we will analyze an upper bound on the expected number of non-defective items that only appear in positive tests. We will then use this to get a high probability upper bound on the number of such items using Markov's Inequality.

Let $T_{1}$ be the number of positive tests. We use the following result from \cite{johnson2018performance} to establish a concentration bound on the value of $T_{1}$. 
\begin{lemma} \emph{\cite[Lemma 1]{johnson2018performance}} \label{lem6}
    Let $\cM \subseteq [n]$ be an arbitrary set of items with $\abs{\cM} = M$, and let $W^{(\cM)}$ be the number of tests in which at least one item of $\cM$ is present. Fix constants $\alpha > 0$ and $\delta > 0$. Then, when drawing $LM$ tests with replacement for items in $\cM$ to be included in, if $LM = \alpha T$ for some $\alpha = \Theta(1)$, then the following holds as $T \to \infty$\emph{:}
    \begin{equation} \label{eq:mcdiarmids}
        \bP(\abs{W^{(\cM)} - (1 - e^{-\alpha})T} > \delta T) \leq 2 \exp\left(-\frac{\delta^{2} T}{\alpha}\right).
    \end{equation}
\end{lemma}

Recalling that $L = \frac{T\log 2}{k}$, we can apply this result with $\cM = S$ and $\alpha = \log 2$, so that $W^{(\cM)}$ corresponds to the number of positive tests.  By doing so, we obtain that $W^{(\cM)} = T_{1}$ concentrates around $\frac{T}{2}$ as follows:
\begin{equation}
    \bP\left( \Big|T_1 - \frac{T}{2}\Big| > \frac{\delta T}{2} \right) \le 2\exp\left(-\frac{\delta^{2}T}{4\log 2}\right), \label{eq:conc_rescaled}
\end{equation}
where we rescaled $\delta$ to $\frac{\delta}{2}$ for later convenience. 
With this established, we can now proceed to bound the probability that a non-defective item only appears in positive tests. Fix some $i \in S^\mathsf{c}$; given $T_{1}$, each test that $i$ is drawn to be included in is positive with probability $\frac{T_{1}}{T}$. Thus, we can bound the probability of $i$ only being included in positive tests as follows:
\begin{align}
   \label{eq:pos_tests_1} \bP(A_{i}) &\coloneqq \bP(i \ \mathrm{only \ appears \ in \ positive \ tests}) \\
   \label{eq:pos_tests_2} &= \sum_{t_1 = 1}^{T}\bP(T_{1} = t_1)\left(\frac{t_1}{T}\right)^{L} \\
   &=\sum_{t_1 = 1}^{ (1+\delta)\frac{T}{2}}\bP(T_{1} = t_1)\left(\frac{t_1}{T}\right)^{L} \nonumber \\
   \label{eq:pos_tests_3}  &\qquad  + \sum_{t_1 = (1+\delta)\frac{T}{2} + 1}^T\bP(T_{1} = t_1)\left(\frac{t_1}{T}\right)^{L} \\
   \label{eq:pos_tests_4} &\leq \left(\frac{(1+\delta)T}{2T}\right)^{L} + 2\exp\left(-\frac{\delta^{2}T}{4\log 2}\right) \\
   \label{eq:pos_tests_5} &= \left(\frac{1 + \delta}{2}\right)^{\frac{T\log 2}{k}} + 2\exp\left(-\frac{\delta^{2}T}{4\log 2}\right) \\
   \label{eq:pos_tests_6} &= \exp\left(-\frac{T\log 2}{k}  \log \frac{2}{1 + \delta}\right) + 2\exp\left(-\frac{\delta^{2}T}{4\log 2}\right)
\end{align}
where \eqref{eq:pos_tests_4} follows from \eqref{eq:conc_rescaled} (and by bounding $\left(\frac{t_1}{T}\right)^{L}$ above by $\big(\frac{(1+\delta )T}{2T}\big)^{L}$ and $1$ in the first and second sums respectively), and \eqref{eq:pos_tests_5} follows from $L = \frac{T \log 2}{k}$.  Taking $T = \frac{1+\delta}{\log 2\cdot\log \frac{2}{1+\delta}}k \log \frac{n}{k} = \big(\frac{1}{\log^{2}2}k \log \frac{n}{k}\big)(1+\epsilon)$ for some arbitrarily small $\epsilon > 0$ (depending on arbitrarily small $\delta > 0$), we obtain the following bound on $\bP(A_{i})$:
\begin{align}
    \label{eq:pos_tests_8} \bP(A_{i}) &\leq \exp\left(-(1+\delta) \log \frac{n}{k}\right) + 2\exp\left(-\frac{\delta^{2}T}{4\log 2}\right) \\
    \label{eq:pos_tests_9} &= \left(\frac{k}{n}\right)^{1+\delta} + 2\exp\left(-\frac{\delta^{2}T}{4\log 2}\right) \\
    \label{eq:pos_tests_10} &\leq 2\left(\frac{k}{n}\right)^{1+\delta},
\end{align}
where \eqref{eq:pos_tests_10} holds for sufficiently large $n$. Therefore, denoting the number of non-defective items that only appear in positive tests as $H$, it follows that $\bE H \leq 2n \left(\frac{k}{n}\right)^{1+\delta} = 2k\left(\frac{k}{n}\right)^{\delta}$. By Markov's inequality, we can then proceed to obtain a high probability bound on $H$:
\begin{align}
    \label{eq:markov_comp_1} P_{\mathrm{e}}^{+} &= \bP(H \geq \eta k) \\
    \label{eq:markov_comp_2} &\leq \frac{\bE H}{\eta k} \\
    \label{eq:markov_comp_3} &\leq \frac{2k \left(\frac{k}{n}\right)^{\delta}}{\eta k} \\
    \label{eq:markov_comp_4} &= k^{o(1)}\left(\frac{k}{n}\right)^{\delta} \\
    \label{eq:markov_comp_5} &= o(1),
\end{align}
where \eqref{eq:markov_comp_4} applies $\eta = k^{-o(1)}$, and \eqref{eq:markov_comp_5} applies $k = \Theta(n^{\theta})$ with $\theta \in (0,1)$. Finally, we note that $T = \big(\frac{1}{\log^{2}2}k \log\frac{n}{k}\big)(1+\epsilon)$ being sufficient for $P_{\mathrm{e}}^{+} \to 0$ implies that a rate of $\log 2 $ is achievable for any $\theta \in (0,1)$. \qed

\subsubsection{Achievability for SUBSET Using DD}
Before discussing the proof of the achievability, we first state the DD algorithm \cite{aldridge2014group}:
\begin{enumerate}
    \item Mark any item that appears in a negative test as non-defective, and every other item (only appearing in positive tests) as a Possible Defective (PD) item. 
    \item Mark all PD items that appear in some test as the only such PD item as defective. 
    \item Mark all remaining items as non-defective and return the set of marked defective items.
\end{enumerate}
Clearly, steps $1$ and $2$ make no mistakes, so any mistakes from the algorithm must be due to step $3$. In other words, any misclassification by the algorithm will be due to a defective item being marked as non-defective, implying that the estimate $\hat{S}_{\mathrm{DD}}$ returned by the algorithm satisfies $\hat{S}_{\mathrm{DD}} \subseteq S$. Thus, proving that the DD algorithm succeeds for SUBSET is equivalent to showing that the number of these false negative items is at most $\eta k$. To prove this for rates less than $\log 2$ (or equivalently for $T \geq \big(\frac{1}{\log 2}k \log_{2}\frac{n}{k}\big)(1+\epsilon)$), we follow the steps of \cite{johnson2018performance}. In particular we wish to show that 
\begin{equation} \label{eq:dd_error}
    \sum_{i \in S}\bP(L_{i} = 0) = o(\eta k),
\end{equation}
where $L_{i}$ is the number of tests in which $i$ appears and no other PD items appear. If we can show that this is the case, then this implies that the number of these false positive items is at most $\eta k$ with probability $1-o(1)$ due to Markov's inequality (since the above sum is the expected number of such items).

To show \eqref{eq:dd_error}, we first introduce some notation\footnote{Due to notational clashes we use slightly different notation compared to \cite{johnson2018performance}, which denotes $\Phi$ as $\beta$.} from \cite{johnson2018performance} (again adopting the near-constant column weight design with $L = \frac{T\log 2}{k}$):
\begin{itemize}
    \item $\delta > 0$ is an arbitrarily small constant to be determined later;
    \item $g^{*} \coloneqq n(\frac{1}{2} + \delta)^{L}$;
    \item $w_{-} \coloneqq \frac{T(1- \delta)}{2}$;
    \item $\Phi \coloneqq \log (\frac{g^{*}L}{w_{-}})$;
    \item $C(L, w_{-}) \coloneqq \exp(\frac{L^{2}}{4w_{-}})$.
\end{itemize}
We now proceed to establish \eqref{eq:dd_error}. We use the following bound derived in \cite[Eqn. 37]{johnson2018performance} as a starting point:
\begin{align}
   \label{eq:dd_startpoint_1} \Lambda &\coloneqq \sum_{i \in S} \bP(L_{i}= 0) \\
   \label{eq:dd_startpoint_2} &\leq \frac{kC(L, w_{-})}{2^{L}} (1 + \delta + e^\Phi(1-\delta))^L + o(\eta k)  \\
   \label{eq:dd_startpoint_3} &\leq \frac{kC(L, w_{-})}{2^{L}} \exp\left((\delta + e^{\Phi}(1 - \delta))L\right) + o(\eta k),
\end{align}
where in \eqref{eq:dd_startpoint_2} the $o(\eta k)$ term consists of terms that are exponentially decaying\footnote{The terms in question are the last two terms in \cite[Eqn. 30]{johnson2018performance}. We refer the interested reader to \cite[Eqns. 40-42]{johnson2018performance} and the set of equations above them for explicit bounds on these terms that establish their exponential decay.} in $n$ and $k$ (and therefore $o(\eta k)$ for the range of $\eta$ we consider), and \eqref{eq:dd_startpoint_3} uses $1+x \leq e^x$. Using the definitions of $g^{*}$ and $w_{-}$ along with $L = \frac{T\log 2}{k}$, we can express $\Phi$ as
\begin{align}
    \label{eq:phi_bound_1} \Phi &= \log\left(\frac{n}{k}\cdot \frac{2\log 2}{1-\delta}\Big(\frac{1}{2}+\delta\Big)^{L}\right) \\
    \label{eq:phi_bound_2} &= \log \frac{n}{k} + L\log\Big(\frac{1}{2} + \delta\Big) + \log\left(\frac{2\log 2}{1-\delta}\right) \\
    \label{eq:phi_bound_3} &= \log \frac{n}{k} - \frac{T\log 2 \cdot \log\Big(\frac{2}{1 + 2\delta}\Big)}{k} + \log\left(\frac{2\log 2}{1-\delta}\right) \\
    \label{eq:phi_bound_4} &= \log\frac{n}{k} - (1+\xi)\log \frac{n}{k} + \log\left(\frac{2\log 2}{1-\delta}\right) \\
    \label{eq:phi_bound_5} &= -\xi\log \frac{n}{k} + \log\left(\frac{2\log 2}{1 -\delta}\right),
\end{align}
where \eqref{eq:phi_bound_3} follows from $L = \frac{T \log 2}{k}$, and \eqref{eq:phi_bound_4} follows by setting $T = \big(\frac{1}{\log2\cdot \log\frac{2}{1+2\delta}}k\log \frac{n}{k}\big)(1+\xi)$ for arbitrarily small $\xi > 0$ and choosing $\delta$ to be such that this equals $\big(\frac{1}{\log^{2}2}k\log \frac{n}{k}\big)(1+\epsilon)$ for arbitrarily small $\epsilon> 0$.
This then implies that $e^{\Phi} = \frac{2\log 2}{1 -\delta}(\frac{n}{k})^{-\xi} \to 0$.
We can then substitute this back into the upper bound on $\Lambda$ in \eqref{eq:dd_startpoint_3}:
\begin{align}
   &\Lambda -o(\eta k)\nonumber \\
   \label{eq:lambda_bound_1} &\leq \frac{kC(L, w_{-})}{2^{L}}\exp\Big(\Big(\delta + 2\log 2 \cdot \Big(\frac{n}{k}\Big)^{-\xi}\Big)L\Big) \\
   &= kC(L, w_{-})\exp\Bigg(- L\log2 \nonumber \\
   \label{eq:lambda_bound_2} &\qquad \qquad \qquad \quad  +\Big(\delta + 2\log 2 \cdot \Big(\frac{n}{k}\Big)^{-\xi}\Big)L\Bigg) \\
   &= kC(L, w_{-})\exp\Bigg(- \frac{T\log^{2}2}{k} \nonumber \\
   \label{eq:lambda_bound_3} & \qquad \quad \quad+ \Big(\delta+2\log 2\cdot \Big(\frac{n}{k}\Big)^{-\xi}\Big)\frac{T\log 2}{k}\Bigg) \\
   \label{eq:lambda_bound_4} &\leq kC(L, w_{-})\exp\left(-\frac{T \log^{2}2}{k} + 2\delta\frac{T\log 2}{k}\right) \\
   \label{eq:lambda_bound_5} &= kC(L, w_{-}) \exp\left(-\Big(1 - \frac{2\delta}{\log 2}\Big)(1+\epsilon)\log\frac{n}{k}\right) \\
   \label{eq:lambda_bound_6} &= kC(L, w_{-})\left(\frac{k}{n}\right)^{(1+\epsilon)(1 - \frac{2\delta}{\log 2})},
\end{align}
where \eqref{eq:lambda_bound_3} follows from the choice of $L$, \eqref{eq:lambda_bound_4} holds for sufficiently large $n$, and \eqref{eq:lambda_bound_5} follows by substituting $T = \big(\frac{1}{\log^{2}2}k\log \frac{n}{k}\big)(1+\epsilon)$. Finally, we note that $C(L, w_{-}) \coloneqq \exp\big(\frac{L^{2}}{4w_{-}}\big) \to 1$ since $L= \Theta(\log n)$ and $w_{-} = \Theta(k\log n)$, meaning we can crudely upper bound $C(L,w_{-})$ by 2 for sufficiently large $n$. Thus, we get the following final upper bound for $\Lambda$:
\begin{align}
    \label{eq:lambda_final_bound_1} \Lambda &\leq 2 k\left(\frac{k}{n}\right)^{(1+\epsilon)(1 - \frac{2\delta}{\log 2})} + o(\eta k) \\
    \label{eq:lambda_final_bound_2} &=o(\eta k),
\end{align}
where \eqref{eq:lambda_final_bound_2} holds by the assumption that $\eta = \Omega((\frac{k}{n})^\lambda)$ for some $\lambda \in [0,1)$, provided that $\delta$ is sufficiently small. Denoting the number of items $i \in S$ such that $L_{i} = 0$ by $D$, it holds that $\bE D = \Lambda$. Using the aforementioned bound on $\Lambda$ and Markov's inequality, we thus have the following high probability bound on there being $\eta k$ or more such items:
\begin{align}
    \label{eq:markov_dd_1} P_{\mathrm{e}}^{-} &= \bP(D \geq \eta k) = o(1).
\end{align}
Additionally, $T = \big(\frac{1}{\log^{2}2}k \log \frac{n}{k}\big)(1+\epsilon)$ tests being sufficient for this to be the case implies that a rate of $\log 2$ is achievable for SUBSET, as desired.  \qed

\subsection{Optimal Results for Asymmetric Approximate Recovery} \label{sec:asymmetric}

In this appendix, we consider the problem of recovering an estimate $\hat{S}$ that contains at most $\alpha_{1}k$ false negatives and $\alpha_{2}k$ false positives for two fixed and possibly distinct parameters $\alpha_{1}, \alpha_{2} \in (0,1)$. While we are not aware of any works giving explicit bounds on the number of tests necessary and sufficient to solve such a problem, we will show that both the achievability and converse bounds follow almost immediately from existing works, with the converse bound being a consequence of \cite[Theorem 1]{scarlett2017how} and the achievability bound following from \cite[Theorem 1]{truong2020all}.

We start with the converse.  Under the problem setup in \cite{scarlett2017how}, the decoder outputs a ``list'' $\cL \subseteq [n]$ with some pre-specified size $L^{*} \geq k$, and the recovery criterion is $\abs{\cL \cap S} \geq (1-\alpha_{1})k$ for some fixed $\alpha_{1} > 0$. The existing converse bound for this problem is stated as follows.

\begin{theorem} \label{thm4} \emph{\cite[Theorem 1]{scarlett2017how}} Fix $\alpha_{1} \in (0,1)$ and suppose that the decoder's list size is $L^{*} \geq k$. Then in order to have $P_{\mathrm{e}}(L^{*}, \alpha_{1}) \coloneqq \bP(\abs{\cL \cap S} < (1-\alpha_{1})k) \not \to 1$ as $n \to \infty$, we require that 
\begin{equation} \label{eq:list_decoding_converse}
    T \geq (1 - \alpha_{1})\Big(k\log_{2}\frac{n}{L^{*}}\Big)\big(1-o(1)\big).
\end{equation}   
\end{theorem}

We claim that this converse implies a converse for the asymmetric approximate recovery problem. To see this, we show that achievability in the asymmetric approximate recovery problem implies achievability in the ``list approximate recovery'' problem (as specified in the previous paragraph), which means that an impossibility result in the list approximate recovery setting is also valid for asymmetric recovery (by taking the contrapositive).  

Consider an estimate $\hat{S}$ such that $\hat{S}$ has at most $\alpha_1k$ false negatives and $\alpha_2k$ false positives. This implies that $(1-\alpha_1)k \leq\abs{\hat{S}} \leq (1+\alpha_2)k$, since one of the bounds on the number of false positives or negatives would trivially be violated if $\abs{\hat{S}}$ were outside this range. We can then extend this estimate $\hat{S}$ to a list $\cL$ such that $\abs{\cL} = (1+\alpha_2)k$ by adding in arbitrary items. Since there are at most $\alpha_1k$ false negatives in $\hat{S}$ by assumption, the same holds for $\cL$ since adding items can only decrease the number of false negatives.  Hence, $\abs{\cL \cap S} \geq (1-\alpha_{1})k$, meaning that achievability in the asymmetric approximate recovery setting implies achievability in the list approximate recovery setting. By Theorem \ref{thm4}, if $T \leq (1-\alpha_1)\big(k\log_2 \frac{n}{k}\big)(1-\epsilon)$ for any fixed, arbitrarily small $\epsilon > 0$, then $P_{\mathrm{e}}(L^*, \alpha_{1}) \to 1$. Hence, if $T \leq (1-\alpha_1)\big(k\log_2 \frac{n}{k}\big)(1-\epsilon)$, the error probability for asymmetric approximate recovery also tends to $1$.

The achievability bound follows from a result in \cite{truong2020all}, which was stated earlier as Lemma \ref{lem7}.  While the result is stated for $\alpha_1 = \alpha_2$, the proof (which we outlined after Lemma \ref{lem7}) reveals that $T = (1-\alpha_{1})\big(k\log_{2}\frac{n}{k}\big)(1+\epsilon)$ tests suffice to recover an estimate $\hat{S}$ of $S$ with at most $\alpha_{1}k$ false negatives and $\xi k$ false positives for arbitrarily small $\xi > 0$. Taking $\xi$ sufficiently small so that $\xi < \alpha_{2}$ yields the desired result. 

Combining the achievability and converse results established above, we conclude that for any positive constants $\alpha_1,\alpha_2$, the optimal threshold for asymmetric approximate recovery is $(1-\alpha_{1})\big(k\log_{2}\frac{n}{k}\big)$.


\section*{Acknowledgment} This work is supported by the National University of Singapore under the Presidential Young Professorship grant scheme. The authors would like to thank Sidharth Jaggi for helpful discussions surrounding the asymmetric approximate recovery problem.

\bibliography{ref.bib}
\bibliographystyle{IEEEtran}

\begin{IEEEbiographynophoto}{Daniel McMorrow} received the B.Sci.~degree in mathematics and the B.Sci.~degree in economics from the University of Bristol in 2024. He is now an incoming PhD student in the Department of Computer Science at the National University of Singapore (NUS). His research interests are in the areas of statistical learning, information theory, and theoretical computer science.
\end{IEEEbiographynophoto}

\begin{IEEEbiographynophoto}{Jonathan Scarlett} received the B.Eng.~degree in electrical engineering and the B.Sci.~degree in computer science from the University of Melbourne, Australia. From October 2011 to August 2014, he was a Ph.D. student in the Signal Processing and Communications Group at the University of Cambridge, United Kingdom. From September 2014 to September 2017, he was post-doctoral researcher with the Laboratory for Information and Inference Systems at the \'Ecole Polytechnique F\'ed\'erale de Lausanne, Switzerland. From January 2018 to April 2024, he was an assistant professor in the Department of Computer Science and Department of Mathematics, National University of Singapore. From April 2024 onwards, he has been an associate professor in the Department of Computer Science and the Department of Mathematics, National University of Singapore. His research interests are in the areas of information theory, machine learning, signal processing, and high-dimensional statistics. He received the Singapore National Research Foundation (NRF) fellowship, and the NUS Presidential Young Professorship award.
\end{IEEEbiographynophoto}


\end{document}